\documentclass[dvipsnames,letterpaper,11pt]{article}
\usepackage[margin=1in]{geometry}
\usepackage[utf8]{inputenc}
\usepackage{graphicx,fullpage,paralist}
\usepackage{amsmath,amssymb,amsthm}
\usepackage{dsfont}
\usepackage{bm}
\usepackage{array}
\newcolumntype{C}{>{$}c<{$}}
\usepackage{hyperref}
\usepackage{booktabs}
\usepackage{subcaption}
\usepackage[font={footnotesize}]{caption}
\usepackage{mathtools}
\usepackage{thmtools}
\usepackage{varwidth}
\usepackage{xstring}
\usepackage{enumitem}
\setlist[itemize]{nosep}
\setlist[enumerate]{nosep}
\usepackage{blkarray}
\usepackage{csquotes}
\usepackage{stmaryrd}
\usepackage[longnamesfirst]{natbib}
\usepackage{doi}

\newcommand{\RR}{\mathbb{R}}
\newcommand{\NN}{\mathbb{N}}
\newcommand{\calA}{\mathcal{A}}
\newcommand{\calC}{\mathcal{C}}
\newcommand{\calV}{\mathcal{V}}
\newcommand{\lis}{\boldsymbol\ell}
\newcommand{\dis}{\mathbf{d}}
\newcommand{\gen}{\mathbf{g}}

\newtheorem{theorem}{Theorem}
\newtheorem{lemma}{Lemma}

\usepackage{multirow}
\usepackage{xcolor}
\hypersetup{
	colorlinks,
	linkcolor={red!50!black},
	citecolor={blue!50!black},
	urlcolor={blue!80!black}
}

\vfuzz \hfuzz
\begin{document}
	
\title{Proportionality in Multiple Dimensions to Design Electoral Systems
\vspace{.3cm}
}
	\author{Javier Cembrano
	\thanks{Department of Algorithms and Complexity, Max Planck Institute for Informatics. {\tt jcembran@mpi-inf.mpg.de}}
	\and Jos\'e Correa			
	\thanks{Department of Industrial Engineering, Universidad de Chile. {\tt correa@uchile.cl}}
	\and Gonzalo D\'iaz			
	\thanks{SC Johnson Graduate School of Management, Cornell University. {\tt gd387@cornell.edu}}
	\and Victor Verdugo
	\thanks{Institute for Mathematical and Computational Engineering, PUC Chile. {\tt victor.verdugo@uc.cl}}
        \thanks{Department of Industrial and Systems Engineering, School of Engineering, PUC Chile.}
	}

\date{}
\maketitle

\begin{abstract}
    How to elect the representatives in legislative bodies is a question that every modern democracy has to answer. This design task has to consider various elements so as to fulfill the citizens' expectations and contribute to the maintenance of a healthy democracy. The notion of proportionality, in that the support of a given idea in the house should be nearly proportional to its support in the general public, lies at the core of this design task. In the last decades, demographic aspects beyond political support have been incorporated by requiring that they are also fairly represented in the body, giving rise to a multidimensional version of the apportionment problem. In this work, we provide an axiomatic justification for a recently proposed notion of multidimensional proportionality and extend it to encompass two relevant constraints often used in electoral systems: a threshold on the number of votes that a list needs in order to be eligible and the election of the most-voted candidate in each district. We then build upon these results to design methods based on multidimensional proportionality. We use the Chilean Constitutional Convention election (May 15-16, 2021) results as a testing ground---where the dimensions are given by political lists, districts, and genders---and compare the apportionment obtained under each method according to three criteria: proportionality, representativeness, and voting power. While local and global methods exhibit a natural trade-off between local and global proportionality, including the election of most-voted candidates on top of methods based on $3$-dimensional proportionality allows us to incorporate both notions while ensuring higher levels of representativeness and a balanced voting power.
\end{abstract}

\thispagestyle{empty}
\newpage
\setcounter{page}{1}

\section{Introduction}

The idea of proportionality is fundamental in mathematics and has been present in a wide variety of applications in many different disciplines.
We usually say that two sets of values are proportional if one can be obtained by scaling the other by a common value.
However, proportionality can be easily extended to the case where more than one dimension is involved, giving rise to the so-called \textit{tensor scaling}. 
This has been widely studied, achieving both a theoretical understanding of it and efficient algorithms for performing it (see, e.g., \cite{allen2017much, deming1940ras, rote2007matrix}).
Applications of this progress are varied and include, among others, transportation planning \citep{EVANS197019}, contingency table analysis \citep{deming1940ras}, national account systems in economies \citep{stone1964multiple}, teletraffic data transport \citep{kruithof1937telefoonverkeersrekening, krupp1979projection} and stochastic matrices \citep{sinkhorn1964relationship}. 
Tensor scaling not only naturally extends the usual notion of proportionality, but also has very strong and natural properties that one would expect of such procedure, as it has been observed by \citet{aumann1985game}, \citet{balinski1989axiomatic}, and \citet{young1987dividing}, among others.

For several applications, an additional property of the solution turns out to be relevant: integrality.
Think, for example, of the case where some population is simultaneously split according to some demographic features, and a certain amount of people among each group is to be selected ``proportionally'' to the relative presence of such group within the total population.\footnote{For instance, this could be done with the purpose of performing some poll or census.}
Another example, that motivates this work and will be developed in further detail, is the election of a deliberative assembly, where the candidates participate as part of a given list in a particular district and belong to a demographic group that is to be proportionally represented in the assembly. 
In these situations, we want to obtain integer values proportional to some input quantities.
The challenge of extending proportionality to this setting gives rise to the so-called divisor methods, initially introduced in the 18th century for the elections of the recently created U.S. House of Representatives and widely used in elections until today.
This notion of proportionality, based on scaling the input numbers by a common value and rounding the result, was extended to the case of two dimensions by \citet{balinski1989axiomatic}.
In recent work, \citet{cembrano2022multidimensional} showed that it is possible to further extend the same approach to the case of arbitrary dimension, as long as we allow some constant deviations on the total number assigned to each group.
These methods use the same approach as divisor methods and constitute a natural way of adapting tensor scaling to the case where integer values are required.

In this work, we focus on the applications of multidimensional proportionality to elections of representatives, in the so-called apportionment problem.
The main motivation is the growing need to incorporate dimensions beyond the classic political and geographical aspects in elections, as modern societies become more complex. 
The increasing need for such methods has driven increasing attention to the subject in the fields of social choice, mathematics, and computer science \citep{flanigan2021fair, lang2018multi, pukelsheim2017proportional}.
In practice, New Zealand's parliament has included ethnic representation for more than 50 years, while Bosnia and Herzegovina's Parliament was proposed to include a division of three types of ``Constituent People'': Bosniacs, Croats, and others \citep{demange2013allocating}. 
A recent case, which we introduce in more detail below and use as a running example throughout the paper, is provided by the election of the Chilean Constitutional Convention in 2021.

After a civil unrest broke out in Chile on October 18, 2019, a political agreement to write a new political constitution was achieved. The agreement was reaffirmed by a referendum approved with 78\% of the votes on October 25, 2020, and the voting process to elect the members of the upcoming Constitutional Convention took place on May 15-16, 2021.
This election drew particular attention due to two constraints on the representative body to be elected. 
First, 17 of the 155 seats of the convention were reserved for 10 different ethnic groups.\footnote{In Chile, around two million people identify themselves as part of these groups, according to the 2017 census.}
Second, a remarkable feature of the election was the incorporation of gender balance in the convention as a way to ensure the proportional representation of women, historically underrepresented in Chilean politics: In the 2021 Parliament election, they achieved their maximum historical representation by taking 34\% of the seats in the lower house. 
This is a reflection of a global situation: The average percentage of women in parliaments around the world has increased from 11.3\% to 26.9\% between 1995 and 2023 according to the \citet{ipu2024} which is still far from a balanced representation. In response to this, various nations like Germany, Belgium, Spain, and more recently, Chile, have tried to promote gender parity in the legislative power through different correction mechanisms, mostly regarding party-level candidate quotas, both in proportional and majoritarian electoral methods.

\paragraph{Our contribution.}
In this paper, we study and develop methods to handle proportionality in multiple dimensions when electing a house of representatives. These methods build upon the notion of multidimensional proportionality introduced by \citet{cembrano2022multidimensional}, which naturally extends that for one and two dimensions \citep{balinski2010fair,balinski1989axiomatic} and is based on scaling each element of each dimension by a multiplier and rounding the result. We begin by showing two results that support the use of this notion of multidimensional proportionality in real-life elections. On the one hand, we prove that it is the (integral version of) the unique notion of multidimensional proportionality that satisfies appropriate versions of three widely-studied axioms in the social choice and apportionment literature: exactness, consistency, and homogeneity. This result extends that of \citet{balinski1989axiomatic} for the $2$-dimensional case. On the other hand, we show that it is possible to extend this notion of multidimensional proportionality to the case where there are lower or upper bounds on the number of seats a specific tuple of elements should receive. These bounds may come, for instance, from the fact that a party must receive at least one seat in a certain district due to having the top-voted candidate, or no seats due to not reaching a certain threshold of votes. Similarly to \citet{cembrano2022multidimensional}, we show that the existence of proportional apportionments deviating up to a constant $\alpha_i$ from the prespecified number of seats in each dimension $i\in \{1,\ldots,d\}$ is guaranteed, provided the feasibility of a linear program and the inequality $\sum_{i=1}^d {1}/{(\alpha_i+2)} \leq 1$.

We then exploit the aforementioned results to propose and study several electoral methods based on $3$-dimensional proportionality (abbreviated as $3$-proportionality in what follows). They are defined in the context of the election of the Chilean Constitutional Convention, so that the dimensions correspond to political lists, where each list should receive a number of seats proportional to its votes, geographic distribution, where each district should receive a number of seats prespecified by law, and gender, where each gender should receive half of the seats. However, we remark that our methods can be directly applied in other contexts.
Besides the method obtained by directly applying the notion of multidimensional proportionality of \citet{cembrano2022multidimensional}, we introduce two additional features in light of our theoretical results: a threshold over the votes obtained by a list in order to be eligible and the election of the top-voted candidate of each district (also known as \emph{plurality} constraints). These features aim to limit the segmentation of political parties and to incorporate a majoritarian component in proportional representation methods, respectively, and are often used together in countries such as Germany, New Zealand, the Republic of Korea, and Mexico. These countries used thresholds of either 3\% or 5\%; we employ the former in our simulations.
As relevant findings, we observe that deviations from the prescribed number of seats in $3$-proportional methods are significantly smaller than the worst-case theoretical guarantees. 
In addition, $3$-proportional methods pose a fundamental trade-off between local and global proportionality and induce a political distribution closer to the well-known principle of \emph{one person, one vote} compared to methods that operate locally. 
Overall, the election of top-voted candidates combined with a threshold over the votes obtained by a list appears as a reasonable midpoint while ensuring high levels of representativeness, in terms of the average number of votes of the elected candidates.
 
\paragraph{Further related work.}
Tensor scaling, in particular the two-dimensional version, has been widely studied since the seminal work by \citet{sinkhorn1964relationship}; some relevant works are those by \citet{sinkhorn1967concerning}, \citet{rothblum1989scalings}, and \citet{nemirovski1999complexity}. Recent works improving the analysis of known algorithms and developing faster ones include those by \citet{allen2017much, chakrabarty2020better, cohen2017matrix}. Particularly relevant for our framework, where integrality is required, are the axiomatic and algorithmic results by \citet{balinski1989algorithms, balinski1989axiomatic}; their algorithm was then analyzed by \citet{kalantari2008complexity}. For an overview of this area, we refer the reader to the survey by \citet{idel2016review}.

The apportionment problem, and in particular divisor methods, have been extensively studied from disciplines such as operations research, computer science, and political science; the books by \citet{balinski2010fair} and \citet{pukelsheim2017proportional} provide historical and theoretical surveys of this topic.
In their seminal work, \citet{balinski1989algorithms, balinski1989axiomatic} extended the notion of proportionality and divisor methods to the case in which the apportionment is ruled by two dimensions, e.g., political parties and geographic districts, studying this extension from an axiomatic and algorithmic point of view. \citet{balinski2008fair}, as well as \citet{balinski1997mexican},
proposed variants of this method, while \citet{maier2010divisor} conducted a real-life benchmark study of biproportional apportionment and its variants, including the election of top-voted candidates in each district. \citet{rote2007matrix} and \citet{gaffke2008divisor,gaffke2008vector} studied biproportional apoortionment via network flow formulations, a tool studied in more general electoral settings by
\citet{pukelsheim2012network}. 
Recently, \citet{cembrano2022multidimensional} extended the biproportional approach to an arbitrary number of dimensions using tools from discrepancy theory, showing that, although proportional apportionments do not always exist, they do if small deviations are allowed. Other recent works regarding the use of tools from optimization and rounding in the development of (randomized) apportionment methods include those by \citet{golz2022apportionment} and \citet{correa2024monotone}.
New perspectives on social choice and apportionment methods have also been studied by \citet{balinski2007theory, balinski2011majority} and \citet{serafini2020mathematics}.

\section{Preliminaries}
\label{sec:prelims}

In the one-dimensional apportionment problem, we are given a vector $\calV\in \NN^n$ containing the votes obtained by each party $i\in \{1,\ldots,n\}$ and a house size $H\in \NN$. The goal is to find an allocation of seats to parties that is proportional to the votes. We let $f^1(\calV, H)=\lambda\calV$ denote the \textit{$1$-dimensional fair share} for input $(\calV,H)$, where $\lambda= H/\sum_{i=1}^{n} \calV_i$. 
Based on divisor methods---which constitute our notion of integral proportionality in what follows---we let $\calA^1(\calV,H)$ be the set of \textit{\mbox{$1$-dimensional} proportional apportionments}, which are the integral vectors obtained by scaling $\calV$ by a multiplier and rounding, i.e., the set of vectors $x\in \NN^n$ satisfying
\begin{align}
    x_i & \in \llbracket \lambda\calV_i \rrbracket \qquad \text{for every } i\in \{1,\ldots,n\} \text{ and}\label{eq:prop-1d}\\
    \sum_{i=1}^n x_i & = H \label{eq:house-1d}
\end{align}
for some $\lambda>0$,
where $\llbracket\cdot \rrbracket$ represents a rounding rule. Throughout this paper, we focus for the sake of simplicity on the Jefferson/D'Hondt method and its generalizations to multiple dimensions, so that $\llbracket\cdot \rrbracket$ will denote the downwards rounding rule.\footnote{However, our theoretical results are valid for any other rounding rule.} More precisely, we define $\llbracket 0\rrbracket = \{0\}$, $\llbracket t\rrbracket = \{r\}$ when $t\in (r,r+1)$ and $\llbracket t\rrbracket = \{r-1,r\}$ when $t=r>0$.\footnote{This set-valued rounding is needed to ensure the existence of proportional apportionments.}
When vectors $I,U\in \NN^n$ such that $\sum_{i=1}^n I_i \leq H \leq \sum_{i=1}^n U_i$ are given as well, we let $\calA^1(\calV,H,I,U)$ be the set of vectors $x\in \NN^n$ satisfying \eqref{eq:house-1d} and
\begin{align}
    x_i & \in \text{mid}\left( \llbracket \lambda\calV_i \rrbracket \cup \{ I_i, U_i \}\right) \qquad \text{for every } i\in \{1,\ldots,n\}\label{eq:prop-1d-bounds}
\end{align}
for some $\lambda>0$, where $\text{mid}(S)$ represents the median value of $S$, for $S\subset \RR$. Note that, in a slight abuse of notation, fixing $I_i=0$ and $U_i=H$ for each $i$ we have $\calA^1(\calV,H) = \calA^1(\calV,H,I,U)$ for every $\calV\in \NN^n$ and $H\in \NN$.
The described apportionments are guaranteed to exist and can be found through fast combinatorial algorithms or linear programming \citep{balinski2010fair,reitzig2023simple}.

The basic definition of one-dimensional apportionments---without the presence of component-wise bounds---was adapted to two dimensions by \citet{balinski1989algorithms, balinski1989axiomatic} and to three and more dimensions by \citet{cembrano2022multidimensional}; we here introduce the solution concept in full generality and discuss its existence and restrictions later on. In the $d$-dimensional apportionment problem, we are given sets $N_1,\ldots,N_d$ and an integer matrix $\calV\in \NN^{\prod_{i=1}^{d}N_i}$
containing the votes of the candidates of each tuple, 
and we define $E(\calV)=\big\{ e\in \prod_{i=1}^{d}N_{i}: \calV_e>0\big\}$.\footnote{This definition will make the notation easier given the assumption that any tuple $e$ with $\calV_e=0$ must not receive any seat, consistently with any definition of a rounding rule.} We are also given values $m_v\in \NN$ for every $v\in N_{i}$ and every $i\in \{1,\ldots,d\}$ determining the number of seats to be assigned to each element,
values $I_e, U_e\in \NN$ for every $e\in E(\calV)$ restricting the number of seats to be assigned to each tuple,
and a house size $H\in \NN$, satisfying
\begin{align}
    \sum_{e\in E(\calV): e_i=v} I_e \leq m_v & \leq \sum_{e\in E(\calV): e_i=v} U_e \qquad \text{for every } i\in \{1,\ldots,d\} \text{ and } v\in N_i, \text{ and} \label{eq:marginals-bounds}\\
    \sum_{v\in N_i}{m_v} & = H \qquad \text{for every } i\in \{1,\ldots,d\}.\label{eq:marginals-house-size}
\end{align}
As the natural extension of the notion of fractional proportionality to the case of multiple dimensions, we say that a tensor $x$ is a \textit{$d$-dimensional fair share} if, for every $i\in \{1,\ldots,d\}$ and $v\in N_i$, there exists a strictly positive value $\lambda_{v}$ such that the following conditions hold:
\begin{align}
\sum_{e\in E(\calV): e_i=v}{x_e} & = m_v\qquad\text{for every }i\in \{1,\ldots,d\} \text{ and } v\in N_i, \text{ and}\label{eq:triprop1-frac}\\
x_{e}&=  \calV_e \prod_{i=1}^{d}\lambda_{e_i} \qquad \text{for every } e\in E(\calV).\label{eq:triprop2-frac}
\end{align}
The existence and uniqueness of such a tensor are guaranteed by standard convex programming techniques; we denote it by $f^d(\calV,m,H)$.
As in the one-dimensional case, but allowing small deviations $\alpha\in \NN^d$ from the prespecified sums, we say that an integer tensor $x$ is an \textit{$\alpha$-approximate $d$-proportional apportionment} if, for every $i\in \{1,\ldots,d\}$ and $v\in N_i$, there exists a strictly positive value $\lambda_{v}$ such that the following conditions hold:
\begin{align}
m_v-\alpha_i&\le \sum_{e\in E(\calV): e_i=v}{x_e}\le m_v+\alpha_i \qquad\text{for every }i\in \{1,\ldots,d\} \text{ and } v\in N_i,\label{eq:triprop1}\\
x_{e}&\in \text{mid} \left( \left\llbracket \calV_e \prod_{i=1}^{d}\lambda_{e_i} \right\rrbracket \cup \{ I_e, U_e \} \right) \qquad \text{for every } e\in E(\calV).\label{eq:triprop2}
\end{align}
We denote the set of tensors $x$ satisfying \eqref{eq:triprop1}-\eqref{eq:triprop2} by $\calA^d(\calV,m,H,I,U;\alpha)$. If a tensor $x$ satisfies this definition with $\alpha_i=0$ for every $i\in \{1,\ldots,d\}$, we just say that $x$ is a $d$-proportional apportionment for this instance, and denote the set of such tensors simply as $\calA^d(\calV,m,H,I,U)$. Further, when $I_e=0$ and $U_e=H$ for every $e\in E(\calV)$, we omit these arguments as before. \citet[Theorem 4]{cembrano2022multidimensional} proved that, under mild conditions,\footnote{These conditions are given by the feasibility of a linear program similar to the one we introduce in Section \ref{sec:theory}.} an $\alpha$-approximate $d$-proportional apportionment is guaranteed to exist whenever $I_e=0$ and $U_e=H$ for every $e\in E(\calV)$ and $\sum_{i=1}^{d}1/(\alpha_i+2)\leq 1,$ and it can be found by using linear programming techniques. For instance, when $d=3$ the above condition is satisfied when $\alpha_1=\alpha_2=\alpha_3=1$, or when $\alpha_1=0$, $\alpha_2=1$ and $\alpha_3=4$. In Section \ref{sec:theory}, we extend this existence result to the case of arbitrary vectors $I$ and $U$.

We often need to aggregate the components or restrict to a certain entry of a $3$-dimensional tensor in our methods; we write a subindex ``+'' for the former case and the subindex with the entry for the latter. We write a subindex $\cdot$ when a certain dimension is preserved. For example, for a tensor $\calV\in D\times L \times G$ and a fixed $d\in D$, $\calV_{d,\cdot,+}$ represents a vector in $\NN^{|L|}$ whose $\ell$-th component is $\sum_{g\in G}\calV_{d\ell g}$.

\section{Theoretical Results}\label{sec:theory}

In this section, we present our theoretical contribution. It consists of two extensions of previous results that we consider relevant for a better understanding and application of apportionment methods based on multidimensional proportionality. The first one, introduced in Section \ref{subsec:axiomatic}, is a characterization of the multidimensional fair share tensor as the only one satisfying a set of three natural axioms, which generalizes the result by \citet{balinski1989axiomatic} for $d=2$ and supports our notion of multidimensional proportionality. The second one, presented in Section \ref{subsec:plurality}, states the existence of approximate $d$-dimensional apportionments with arbitrary component-wise bounds, thus extending the result by \citet{cembrano2022multidimensional} and providing the basis for the application of our methods defined in Sections \ref{subsubsec:TPP} and \ref{subsubsec:TPP3}.

\subsection{Multidimensional Proportional Allocations: An Axiomatic Approach}
\label{subsec:axiomatic}

In this section, we provide an axiomatic characterization of the $d$-dimensional fair share for arbitrary $d$, thereby justifying the notion of integral proportionality in multiple dimensions we stick to, which can be understood as a rounding of fractional proportional tensors.

Given sets $N_1,\ldots,N_d$, values $m_v$ for every $i\in \{1,\ldots,d\}$ and $v\in N_{i}$, and $H\in \NN$, we define the region of feasible allocations as the set
$$R(m,H) = \Bigg\{f\in  \RR^{\prod_{i=1}^{d}N_{i}}:\sum_{\substack{e\in \prod_{i=1}^{d}N_{i}:e_i=v}} f_{e} = m_v \ \text{ for every }i\in \{1,\ldots,d\} \text{ and } v\in N_i\Bigg\}.$$
An \emph{allocation method} $F$ is a mapping from an instance given by an integer vote matrix $\calV$ with entries in $\prod_{i=1}^{d}N_{i}$, values $m_v$ for every $i\in \{1,\ldots,d\}$ and $v\in N_{i}$, and $H\in \NN$, to a nonempty subset $F(\calV,m,H)$ of $R(m,H)$. We need some additional notation before introducing the axioms. Consider $f\in F(\calV,m,H)$ and subsets $N'_1\subseteq N_1,\ldots, N'_{d}\subseteq N_d$. 
Let $E'=\Pi_{i=1}^{d}{N'_i}$, let $y_{E'}$ be the subtensor of $y$ restricted to $E'$, and let $m'(f)$ be such that $(m'(f))_v$ is defined as follows for each $i\in \{1,\ldots,d\}$ and $v\in N'_i$:
$$(m'(f))_v = \sum_{e\in E':e_i=v}{f_e} = m_v - \sum_{e\not\in E': e_i=v}{f_e}.$$
We now formally introduce three axioms for allocation methods.

\begin{enumerate}[label=(\Roman*)]
    \item {\bf (Exactness)} If there exists $\delta\in \RR$ such that $\delta \calV \in R(m,H)$, then $F(\calV,m,H)=\{ \delta \calV \}$.\label{axiom:exact}
    \item {\bf (Consistency)} If $f\in F(\calV,m,H)$, the following holds: \label{axiom:consistency}
        \begin{enumerate}[label=(\alph*)]
        \item For every family of sets $N'_1\subseteq N_1,\ldots, N'_{d}\subseteq N_d$, we have that
        \[
            f_{E'}\in  F\left(\calV_{E'},m'(f),\sum_{e\in E'} f_{e}\right).
        \]
        \label{item:cons-1}
        \item If $g\in F(\calV_{E'},m'(f),\sum_{e\in E'} f_{e})$, then the matrix defined to be equal to $g$ on $E'$ and $f$ elsewhere belongs to $F(\calV,m,H)$.\label{item:cons-2}
        \end{enumerate}
    \item {\bf (Homogeneity)} Suppose that $i\in \{1,\ldots,d\}$ and $v_1,v_2\in N_i$ are such that $\calV_{e}=\calV_{e'}$ for every pair $e,e'\in E$ with $e_i=v_1,\ e'_i=v_2$, and $e_j=e'_j$ for every $j\not=i$. Furthermore, suppose that $m_{v_1}=m_{v_2}$. Then, for every $f\in F(\calV,m,H)$ it holds $f_{e}=f_{e'}$ for every $e,e'\in E$ with $e_i=v_1,\ e'_i=v_2$, and $e_j=e'_j$ for every $j\not=i$.\label{axiom:homogen}\\
\end{enumerate}
We denote by $F^*$ the allocation method that maps every instance $(\calV,m,H)$ to a singleton containing its fair share tensor $f^d(\calV,m,H)$. 
The following theorem provides an axiomatic characterization of this method, extending that by \citet{balinski1989axiomatic} for the case of $d=2$.

\begin{theorem}
\label{thm:fair-share}
$F^*$ is the unique allocation method satisfying exactness, consistency, and homogeneity.
\end{theorem}

The following lemma is key for proving uniqueness. It states that every allocation method satisfying consistency and homogeneity outputs a single allocation for every instance, and its output is invariant under scaling of the vote tensor in a single dimension.

\begin{lemma}
\label{lem:fs-scaled}
Let $F$ be an allocation method that satisfies consistency and homogeneity, and let $(\calV,m,H)$ be an instance such that $\calV$ has entries in $\prod_{j=1}^{d}N_j$. Let $i\in \{1,\ldots,d\}$ and $\mu_v>0$ for each $v\in N_i$ be arbitrary, and let $\calV'\in\prod_{j=1}^{d}N_j$ be defined as \[\calV'_{e}=\mu_{e_i}\calV_e \qquad \text{for every }\;e\in \prod_{j=1}^{d}N_j.\] 
Then,  $|F(\calV,m,H)|=1$ and $F(\calV,m,H)= F(\calV',m,H)$.
\end{lemma}

\begin{proof}
Let $\calV,~ i,~ \mu_v$ for each $v\in N_i$, and $\calV'$ be defined as in the statement. Let $N'_{i}$ be a copy of $N_i$, so that we denote by $v'\in N'_i$ the copy of $v\in N_i$. Consider the related problem $(\calV'',m'',2H)$, where $\calV''\in \prod_{j=1}^{i-1}N_j\times (N_i\cup N'_i) \times \prod_{j=i+1}^{d}N_j$ is defined as
\[
\calV''_e=
\begin{cases}
\calV_e &\text{if }e_i\in N_i,\\ 
\mu_{e_i}\calV_e &\text{if }e_i\in N'_i, 
\end{cases}
\]
for every $e$, $m''_v=m_v$ for every $j\in\{1,\ldots,d\}$ and $v\in N_j$, and $m''_{v'}=m_v$ for every $v\in N_i$.
Let $E=\prod_{j=1}^{d}N_j$ and $E'=\prod_{j=1}^{i-1}N_j\times N'_i \times \prod_{j=i+1}^{d}N_j$ denote the tuples with the $i$th entry in $N_i$ and $N'_i$, respectively.
Let finally $f\in F(\calV'',m'',2H)$.

By homogeneity, we have that $f_e=f_{e'}$ for every $e,e'$ such that $e_i=v\in N_i,\ e'_i=v'\in N'_i$ and $e_j=e'_j$ for every $j\not=i$.
By consistency, $f_E\in  F(\calV,m,H)$ and $f_{E'}\in  F(\calV',m,H)$ and, moreover, if there was $g_E\in F(\calV,m,H)$ with $g_E\not=f_E$, then the tensor defined to be equal to $g_E$ on $E$ and to $f_{E'}$ on $E'$ would also be in $F(\calV'',m'',2H)$, contradicting homogeneity (and similarly if there was $g_{E'}\in F(\calV',m,H)$ with $g_{E'}\not=f_{E'}$). We conclude that $|F(\calV,m,H)|=1$ and $F(\calV,m,H)= F(\calV',m,H)$.
\end{proof}

We are now ready to prove Theorem \ref{thm:fair-share}.

\begin{proof}[Proof of Theorem \ref{thm:fair-share}]
We first prove that the fair share method $F^*$ satisfies axioms \ref{axiom:exact}-\ref{axiom:homogen}, and then we prove that it is the unique one satisfying them. To prove the former, consider an arbitrary instance $(\calV,m,H)$ and let $f^*= f^d(\calV,m,H)$ denote the fair share tensor for this instance, for ease of notation. Let $\lambda_v$ for each $i\in \{1,\ldots,d\}$ and $v\in N_i$ be multipliers that define $f^*$.\\

\noindent{\bf Exactness.} Assume that $\delta\in \RR$ is such that $\delta \calV \in R(m,H)$. By taking $\lambda_v=\delta^{1/d}$ for every $i\in\{1,\ldots,d\}$ and $v\in N_{i}$, we have $\calV_{e}\prod_{i=1}^{d}\lambda_{e_i} = \delta \calV_{e}$ for each $e\in \prod_{i=1}^{d}{N_{i}}$. Therefore, we conclude from the uniqueness of the fair share tensor that $f^*= \delta \calV$ and thus $F^*(\calV,m,H)=\{\delta\calV\}$.\\

\noindent{\bf Consistency.} Let $N'_1\subseteq N_1,\ldots, N'_d\subseteq N_d$ be arbitrary sets and $E' = \prod_{i=1}^{d}{N'_i}$, and observe that for each $i\in \{1,\ldots,d\}$ and $v\in N'_i$ we have that
\[
    \sum_{\substack{e\in \prod_{j=1}^{d}N'_j:\\e_i=v}}{f^*_{e}} = \sum_{\substack{e\in \prod_{j=1}^{d}N_j:\\e_i=v}}{f^*_{e}}-\sum_{\substack{e\not\in \prod_{j=1}^{d}N'_j:\\e_i=v}}{f^*_{e}} = (m'(f^*))_v.
\]
Since $f^*_{e} = \calV_{e}\prod_{i=1}^{d}{\lambda_{e_i}}$ for every $e\in E'$, we conclude that 
\[
    f^*_{E'}\in  F^*\left(\calV_{E'},m',\sum_{e\in E'} f^*_{e}\right),
\]
which proves part \ref{item:cons-1}.
For part \ref{item:cons-2}, since fair share tensors are unique,  $g\in F^*(\calV_{E'},m',\sum_{e\in E'} f^*_{e})$ implies $g=f^*_{E'}$, so it is straightforward that the matrix defined to be equal to $g$ on $E'$ and $f^*$ elsewhere belongs to $F^*(\calV,m,H)$.\\

\noindent{\bf Homogeneity.} Let $i\in \{1,\ldots,d\}$ and $v_1,v_2\in N_i$ be such that $\calV_{e}=\calV_{e'}$ for every pair $e,e'\in E$ with $e_i=v_1,\ e'_i=v_2$, and $e_j=e'_j$ for every $j\not=i$. Furthermore, suppose that $m_{v_1}=m_{v_2}$. Suppose we had $f^*_{e} \not = f^*_{e'}$ for some $e,e'\in E$ with $e_i=v_1,\ e'_i=v_2$ and $e_j=e'_j$ for every $j\not=i$. This would imply $\lambda_{v_1} \not = \lambda_{v_2}$; suppose w.l.o.g.\ $\lambda_{v_1} > \lambda_{v_2}$. This yields
\begin{align*}
\sum_{\substack{e\in \prod_{j=1}^{d}:\\ e_i=v_1}}f^*_{e} & = \sum_{\substack{e\in \prod_{j=1}^{d}:\\ e_i=v_1}} \left(\calV_{e}\prod_{j=1}^{d}\lambda_{e_j}\right) = \sum_{\substack{e\in \prod_{j=1}^{d}:\\ e_i=v_1}} \left(\calV_{e}\lambda_{v_1}\prod_{j\in \{1,\ldots,d\}\setminus \{i\}}\lambda_{e_j}\right)\\
& > \sum_{\substack{e\in \prod_{j=1}^{d}:\\ e_i=v_2}} \left(\calV_{e}\lambda_{v_2}\prod_{j\in \{1,\ldots,d\}\setminus \{i\}}\lambda_{e_j}\right) = \sum_{\substack{e\in \prod_{i=1}^{d}:\\ e_i=v_2}} \left(\calV_{e}\prod_{j=1}^{d}\lambda_{e_j}\right) = \sum_{\substack{e\in \prod_{i=1}^{d}:\\ e_i=v_2}} f^*_{e},
\end{align*}
where the inequality holds both from $\lambda_{v_1} > \lambda_{v_2}$ and from the fact that $\calV_{e}=\calV_{e'}$ for every pair $e,e'\in E$ with $e_i=v_1,\ e'_i=v_2$, and $e_j=e'_j$ for every $j\not=i$.
This contradicts \eqref{eq:triprop1-frac} in the definition of a fair share tensor since $m_{v_1}=m_{v_2}$, so we conclude that $f^*_{e} = f^*_{e'}$ for every $e,e'\in E$ with $e_i=v_1,\ e'_i=v_2$, and $e_j=e'_j$ for every $j\not=i$.\\

Now, let $F$ be an allocation method that satisfies axioms \ref{axiom:exact}-\ref{axiom:homogen} and let $(\calV,m,H)$ be an arbitrary instance. For every $i\in \{1,\ldots,d\}$ and every $v\in N_i$, let $\lambda_v$ be fair share multipliers, i.e., values such that $f^*$ defined as $f^*_{e} = \calV_{e}\prod_{i=1}^{d}{\lambda_{e_i}}$ for every $e\in \prod_{i=1}^{d}{N_i}$ satisfies \eqref{eq:triprop1-frac}-\eqref{eq:triprop2-frac}.
Define $\calV'$ as $\calV'_e=\calV_{e}\prod_{i=1}^{d}{\lambda_{e_i}}$ for every $e\in \prod_{i=1}^{d}{N_i}$. Applying the result from Lemma \ref{lem:fs-scaled} $d$ times with $\mu_v=\lambda_v$ for every $i\in \{1,\ldots,d\}$ and $v\in N_i$, we obtain that $F(\calV',m,H) = F(\calV,m,H)$. But $\calV'\in R(m,H)$, thus from exactness we have that $F(\calV',m,H)=\{ \calV' \}$. We conclude that $F(\calV,m,H)= \{ \calV' \} = \{ f^*\}$. 
\end{proof}

We remark that the fair share tensor can be found by solving a convex optimization program, whose constraints correspond to $f$ being an element of the polytope $R(m,H)$ and whose convex objective function is
\[ \sum_{e\in \prod_{i=1}^{d}N_i}\left(\log\left(\frac{f_{e}}{\calV_{e}}\right)-1\right).\]
By standard KKT optimality conditions, it can be verified that the optimal solution of such a program is the $d$-dimensional fair share \citep{idel2016review}.

\subsection{Approximate $3$-proportional Apportionments with Plurality}
\label{subsec:plurality}

In this section, we extend the existential and algorithmic results from \citet{cembrano2022multidimensional} for the case of apportionments with lower and upper bounds on each entry. This result will guarantee the correctness of our methods with plurality constraints and may be of independent interest for other methods incorporating bounds on the number of seats that certain entries should be allocated.

In order to formalize this extension, we consider the following integer linear program inspired by the network flow approach introduced by \citet{rote2007matrix} for matrix scaling, and used by \citet{gaffke2008divisor, gaffke2008vector} to model biproportionality.
\begin{align}
\text{minimize}  \quad \;\sum_{e\in E(\calV)} & \sum_{t=1}^{H} y_{e}^t\log\left(\frac{t}{\mathcal{V}_{e}}\right) \label{eq:objective}\\
\text{subject to}  \quad\quad \sum_{t=1}^{H}y_{e}^t & = x_{e} \ \quad \quad \text{ for every }e\in E(\calV),\label{eq:aggregation}\\
\quad \quad \sum_{\substack{e\in E(\calV):\\ e_i=v}}x_{e} & = m_v \quad \quad \text{ for every }i\in \{1,\ldots,d\} \text{ and }v\in N_i,\label{eq:marginals-lp}\\
\quad\quad x_{e} & \geq I_{e} \hspace{0.95cm} \text{ for every }e\in E(\calV),\label{eq:lb-lp}\\
\quad\quad x_{e} & \leq U_{e} \hspace{0.86cm} \text{ for every }e\in E(\calV),\label{eq:ub-lp}\\
\quad \quad  y_{e}^{t} & \in \{0,1\}  \hspace{0.38cm} \text{ for every }e\in E(\calV)\text{ and }t\in \{1,\ldots,H\}.\label{eq:binary}
\end{align}
For each $e\in E(\calV)$ and $t\in\{1,\ldots,H\}$ we have a binary variable $y_{e}^{t}$ whose cost in the objective function is given by $\log(t/\mathcal{V}_{e})$. The variable $x_{e}$ represents the total number of seats to be allocated in the apportionment for candidates of tuple $e$; constraint (\ref{eq:aggregation}) takes care of aggregating the seats in these variables.
Constraint \eqref{eq:marginals-lp} enforces every feasible solution to satisfy the marginals and constraints \eqref{eq:lb-lp} and \eqref{eq:ub-lp} ensure respecting the bounds. The following theorem states the main result for this setting with plurality.

\begin{theorem}\label{thm:plurality}
Let $\calV,\ m$, $H$, $I$, and $U$ be
such that the linear relaxation of (\ref{eq:objective})-(\ref{eq:binary}) is feasible.
Let $\alpha_1,\ldots \alpha_d$ be nonnegative integers such that 
$\sum_{i=1}^{d}1/(\alpha_i+2)\le 1$.
Then, there exists an \mbox{$\alpha$-approximate} $d$-dimensional proportional apportionment.
\end{theorem}

This theorem states the same existence result of \citet[Theorem 4]{cembrano2022multidimensional}. The discrepancy approach \citep[Algorithm 1 and Theorem 5]{cembrano2022multidimensional} remains valid for this setting as well, so the only result we need to extend is their Lemma 5, stating that any rounding of an optimal solution of the linear relaxation satisfies the proportionality condition. In our context, this entails showing that any rounding of an optimal solution of the linear relaxation of \eqref{eq:objective}-\eqref{eq:binary} satisfies the modified proportionality condition \eqref{eq:triprop2}.
In order to do so, we first observe that by duality, for any feasible solution $(x,y)$ of the linear relaxation of (\ref{eq:objective})-(\ref{eq:binary}), we have that $(x,y)$ is optimal if and only if there exists a dual solution $(\Lambda,\omega^-, \omega^+,\beta)$ such that the following conditions hold:
\begingroup
\allowdisplaybreaks
\begin{align}
	\sum_{i=1}^{d}\Lambda_{e_i} + \omega^-_e + \omega^+_e + \beta_{e}^t - \log \left(\frac{t}{\mathcal{V}_{e}}\right) & \le 0 &\text{ for every }e\in E(\mathcal{V}) \text{ and } t\in \{1,\ldots,H\}, \label{eq:dual1}\\
    y_{e}^t\left[\sum_{i=1}^{d}\Lambda_{e_i} + \omega^-_{e} + \omega^+_e + \beta_{e}^t- \log \left(\frac{t}{\mathcal{V}_{e}}\right) \right] & = 0 &\text{ for every }e\in E(\mathcal{V}) \text{ and } t\in \{1,\ldots,H\}, \label{eq:dual2}\\
	\omega^-_{e}\left(\sum_{t=1}^{H}y_{e}^t-I_e\right)&=0& \text{ for every }e\in E(\mathcal{V}), \label{eq:dual3}\\
	\omega^-_{e}&\geq 0 & \text{ for every }e\in E(\mathcal{V}), \label{eq:dual4}\\
    \omega^+_{e}\left(U_e - \sum_{t=1}^{H}y_{e}^t\right)&=0& \text{ for every }e\in E(\mathcal{V}), \label{eq:dual5}\\
	\omega^+_{e}&\leq 0 & \text{ for every }e\in E(\mathcal{V}), \label{eq:dual6}\\
	\beta_{e}^t(y_{e}^t-1)&=0 & \text{ for every }e\in E(\mathcal{V}) \text{ and } t\in \{1,\ldots,H\},& \label{eq:dual7}\\
	\beta_{e}^t&\le 0 & \text{ for every }e\in E(\mathcal{V}) \text{ and } t\in \{1,\ldots,H\}, & \label{eq:dual8}
\end{align}
\endgroup
where $\Lambda_v$ is the dual variable associated to the constraint (\ref{eq:marginals-lp}) for every $i\in \{1,\ldots,d\}$ and $v\in N_i$, $\omega^-_{e}$ is the dual variable associated to the constraint \eqref{eq:lb-lp} for every $e$, $\omega^+_{e}$ is the dual variable associated to the constraint \eqref{eq:ub-lp} for every $e$, and $\beta_{e}^t$ is the dual variable associated to the upper bound of one on the value of $y_{e}^t$ for every $e\in E(\calV)$ and $t\in \{1,\ldots,H\}$.

We refer to a tuple $(x,y,\Lambda,\omega^-,\omega^+,\beta)$ jointly satisfying \eqref{eq:dual1}-\eqref{eq:dual8}, with $(x,y)$ feasible for the linear relaxation of (\ref{eq:objective})-(\ref{eq:binary}), as an optimal primal-dual pair for this linear relaxation. The following lemma establishes that any rounding of an optimal solution $x$ of this linear relaxation satisfies the proportionality condition \eqref{eq:triprop2}. As mentioned, this lemma along with the results of \citet[Theorem 1, Theorem 5]{cembrano2022multidimensional} allows us to conclude Theorem \ref{thm:plurality}. 
We remark that when $d=2$, the feasibility of the linear relaxation is guaranteed when the bounds $I$ and $U$ are consistent with the marginals $m$, i.e., when condition \eqref{eq:marginals-bounds} holds. This can be seen by modeling the apportionment as a network flow \citep{gaffke2008divisor} and applying Hoffman's Circulation Theorem \citep{hoffman1960recent}.
For $d\geq 3$, it is guaranteed when $I_e=0$, $U_e=H$, and all entries of $\calV$ are strictly positive \citep{cembrano2022multidimensional}.

\begin{lemma}
\label{lem:relaxation}
Let $(x,y,\Lambda,\omega^-,\omega^+,\beta)$ be an optimal primal-dual pair for the linear relaxation of \mbox{(\ref{eq:objective})-(\ref{eq:binary})} with input $\calV,\ m,\ H,\ I$, and $U$.
If $\bar{x}$ is an integral vector with entries in $E(\calV)$ such that $\bar x_{e}\in \{\lfloor x_{e}\rfloor,\lceil x_{e}\rceil\}$ for every $e\in E(\calV)$, then defining $\lambda_v=\exp(\Lambda_v)$ for every $i\in \{1,\ldots,d\}$ and $v\in N_i$ we have that
\[
    \bar x_{e} \in \text{mid} \left( \left\llbracket \calV_e \prod_{i=1}^{d}\lambda_{e_i} \right\rrbracket \cup \{ I_e, U_e \} \right) \qquad \text{for every } e\in E(\calV).
\]
\end{lemma}

\begin{proof}
Let $(x,y,\Lambda,\omega^-,\omega^+,\beta)$, $\lambda$, and $\bar{x}$ be as in the statement. For the first part of the proof, we follow \citet{cembrano2022multidimensional}. For every $e\in E(\mathcal{V})$, we consider the value  
\[t_{e}=\max\{t\in \{1,\ldots,H\}:y_{e}^t>0\}\] 
when $x_{e}>0$ and $t_{e}=0$ when $x_{e}=0$, and we claim that for every $e\in E(\calV)$ it holds that $t_{e}=\lceil x_{e}\rceil$. This is straightforward for the tuples $e\in E(\calV)$ with $x_{e}=0$. 
Constraint \eqref{eq:aggregation} further implies that $t_{e}\geq \lceil x_{e} \rceil$. Suppose towards a contradiction that $e\in E(\calV)$ is such that $t_{e}>\lceil x_{e} \rceil$. Then, there are at least two fractional values in $y^1_{e}\ldots,y_{e}^H$; consider in addition the solution $(x,Y)$ constructed as follows: 
$$
Y_{e}^t=
\begin{cases}
1 & \text{for every } t<\lfloor x_{e}\rfloor +1, \\
x_{e}-\lfloor x_{e}\rfloor & \text{for } t=\lfloor x_{e}\rfloor+1,\\
0& \text{ otherwise.} 
\end{cases}
$$
The solution $(x,Y)$ is feasible for the linear relaxation of (\ref{eq:objective})-(\ref{eq:binary}) since 
\[\sum_{t=1}^{H}Y_{e}^t=(\lfloor x_{e}\rfloor+1)-1+x_{e}-\lfloor x_{e}\rfloor=x_{e}.\]
Additionally, since for every $e\in E(\mathcal{V})$ the function $\log(t/\mathcal{V}_{e})$ is strictly increasing as a function of $t\in \{1,\ldots,H\}$, we obtain that $\sum_{t=1}^{H}Y_{e}^t\log(t/\mathcal{V}_{e}) < \sum_{t=1}^{H}y_{e}^t\log(t/\mathcal{V}_{e})$, a contradiction to the optimality of $(x,y)$. This concludes the proof of the claim.

We now let $e\in E(\calV)$ be arbitrary and note that since $x$ satisfies constraint \eqref{eq:lb-lp} we have that $I_e\leq x_{e}\leq U_e$. In particular, if $I_e=U_e$ the condition in the statement holds trivially. We thus consider an arbitrary tuple $e\in E(\calV)$ with $I_e<U_e$ and distinguish three cases; we will conclude the condition in the statement for each of them.

If $x_{e}=I_e$, from the fact that $t_{e}=\lceil x_{e}\rceil$ we have that $y^{I_e+1}_{e}=0$ and then condition $\eqref{eq:dual7}$ implies that $\beta^{I_e+1}_{e}=0$. Moreover, since $I_e<U_e$ we have that $x_e<U_e$ and condition \eqref{eq:dual5} implies that $w^+_e=0$.
Replacing in condition $\eqref{eq:dual1}$ we obtain 
\[
    \log\left(\frac{I_e+1}{\mathcal{V}_{e}}\right) - \sum_{i=1}^{d}\Lambda_{e_i} \geq \omega^-_{e}.
\]
Since $\omega^-_{e}\geq 0$ from \eqref{eq:dual4}, this implies that $\mathcal{V}_{e} \exp(\sum_{i=1}^{d}\Lambda_{e_i}) \leq I_e+1$.
Therefore, we have either $\calV_e \prod_{i=1}^{d}\lambda_{e_i} < I_e$ or $I_e\in \left\llbracket \calV_e \prod_{i=1}^{d}\lambda_{e_i} \right\rrbracket$. In both cases, we conclude that
\[
    \bar x_{e} = x_e = I_e \in \text{mid} \left( \left\llbracket \calV_e \prod_{i=1}^{d}\lambda_{e_i} \right\rrbracket \cup \{ I_e, U_e \} \right).
\]

If $x_{e}=U_e$, from the fact that $t_{e}=\lceil x_{e}\rceil$ we have that $y^{U_e}_{e}=1$ and then condition $\eqref{eq:dual2}$ yields
\[
    \sum_{i=1}^{d}\Lambda_{e_i} + \omega^-_{e} + \omega^+_e + \beta_{e}^{U_e}- \log \left(\frac{U_e}{\mathcal{V}_{e}}\right) = 0.
\]
Moreover, since $I_e<U_e$ we have that $x_e>I_e$ and condition \eqref{eq:dual3} implies that $\omega^-_e=0$.
Replacing in the previous equation, we obtain 
\[
    \log \left(\frac{U_e}{\mathcal{V}_{e}}\right) - \sum_{i=1}^{d}\Lambda_{e_i} = \omega^+_e + \beta_{e}^{U_e} \leq 0,
\]
where we also used that $\omega^+_e\leq 0$ and $\beta^{U_e}_e \leq 0$ due to conditions \eqref{eq:dual6} and \eqref{eq:dual8}, respectively. This implies that $\mathcal{V}_{e} \exp(\sum_{i=1}^{d}\Lambda_{e_i}) \geq U_e$.
Therefore, we have either $\calV_e \prod_{i=1}^{d}\lambda_{e_i} > U_e+1$ or $U_e\in \left\llbracket \calV_e \prod_{i=1}^{d}\lambda_{e_i} \right\rrbracket$. In both cases, we conclude that 
\[
    \bar x_{e} = x_e = U_e \in \text{mid} \left( \left\llbracket \calV_e \prod_{i=1}^{d}\lambda_{e_i} \right\rrbracket \cup \{ I_e, U_e \} \right).
\]

Finally, when $I_e < x_{e} < U_e$, from conditions \eqref{eq:dual3} and \eqref{eq:dual5} we have that $\omega^-_{e}=\omega^+_e=0$. 
Fixing $h=t_{e}=\lceil x_{e}\rceil$, the complementary slackness condition (\ref{eq:dual2}) implies that 
\begin{equation}
    \sum_{i=1}^{d}\Lambda_{e_i} + \beta_{e}^{h}- \log\left(\frac{h}{\mathcal{V}_{e}}\right)=0,\label{eq:comp-slack}
\end{equation}
since $y_{e}^{h}>0$.
By (\ref{eq:dual8}) we have that $\beta_{e}^{h}\le 0$ and therefore we conclude that $\sum_{i=1}^{d}\Lambda_{e_i} \ge \log(h/\mathcal{V}_{e})$, implying $t_{e} =h\le \mathcal{V}_{e}\exp(\sum_{i=1}^{d}\Lambda_{e_i})$.
We know that for every $t> t_{e}$ it holds $y_{e}^t=0$, so the complementary slackness condition (\ref{eq:dual7}) implies that $\beta_{e}^{t}=0$.
Therefore, condition (\ref{eq:dual1}) implies that $\sum_{i=1}^{d}\Lambda_{e_i} \le \log(t/\mathcal{V}_{e})$, which is satisfied in particular for $t=t_{e}+1$.
We conclude that $\mathcal{V}_{e} \exp(\sum_{i=1}^{d}\Lambda_{e_i}) \le t_{e}+1$.
Putting it all together, we have that
\[
\lceil x_{e}\rceil \leq \mathcal{V}_{e} \prod_{i=1}^{d}\lambda_{e_i} \le \lceil x_{e}\rceil +1,
\]
implying $\lceil x_{e}\rceil \in \llbracket \mathcal{V}_{e} \prod_{i=1}^{d}\lambda_{e_i}\rrbracket$. If $x_{e}$ is integer, then $\bar x_{e} = x_{e} = \lceil x_{e}\rceil \in \{I_e+1,\ldots,U_e-1\}$, so we conclude.
If $x_{e}$ is fractional, we have $0<y_{e}^{h}<1$ for $h=t_e$, thus the complementary slackness condition \eqref{eq:dual7} implies that $\beta^{h}_{e}=0$. Therefore, we conclude from \eqref{eq:comp-slack} that $\mathcal{V}_{e}\exp(\sum_{i=1}^{d}\Lambda_{e_i})=h=t_{e}=\lceil x_{e}\rceil$.
When $\bar x_{e}=\lceil x_{e}\rceil$ we have that 
\[\bar x_{e}=\lceil x_{e}\rceil=\mathcal{V}_{e}\prod_{i=1}^{d}\lambda_{e_i}\le \lceil x_{e}\rceil+1=\bar x_{e}+1.\]
Similarly, when $\bar x_{e}=\lfloor x_{e}\rfloor$ we have that 
\[\bar x_{e}+1=\lfloor x_{e}\rfloor +1=\lceil x_{e}\rceil=\mathcal{V}_{e}\prod_{i=1}^{d}\lambda_{e_i}\ge \lfloor x_{e}\rfloor=\bar x_{e}.\]
In both cases we have $\bar x_{e} \in \llbracket \mathcal{V}_{e} \prod_{i=1}^{d}\lambda_{e_i}\rrbracket$, thus we conclude once again.
\end{proof}

\section{Description of the Apportionment Methods}\label{sec:methods}

In the application we consider, there is a set of districts $D$, a set of lists $L$, and a set of genders $G$. There is also a set of candidates $\calC$, and each candidate $c\in \calC$ competes in a district $\dis(c)\in D$, belongs to a list $\lis(c) \in L$ and is of gender $\gen(c)\in G$. The results of the election consist of the votes $v(c)\in \NN$ each candidate $c$ obtains. We are also given a vector $q\in \NN^{D}$ containing the number of seats to be assigned in each district $d\in D$ according to the law, and a strictly positive integer number $H$ corresponding to the total number of seats to be allocated. 
Since we are interested in the aggregated votes when defining algorithms and evaluating their performance, we construct a matrix $\calV\in \NN^{D\times L\times G}$ such that $\calV_{d\ell g}$ is the sum of the votes of the candidates of list $\ell$ and gender $g$ in district $d$, i.e.
\[
\calV_{d \ell g}=\sum_{\substack{c\in \calC:\: \dis(c)=d, \\ \lis(c)=\ell,\, \gen(c)=g}}{v(c)}.
\]
The objective of the methods discussed in this work is to find an assignment of candidates to seats represented by a function $\chi$, so $\chi(c)=1$ if candidate $c$ is elected and $\chi(c)=0$ otherwise. We naturally require that $\sum_{c\in \calC}\chi(c)=H$.

To illustrate how our methods work, we use a small example of an artificial election consisting of six lists from A to F, three districts from D1 to D3, and a house size of $H=33$ seats.\footnote{This data is actually a modification of the results of the 2021 Constitutional Convention election in Chile, where only the six most-voted lists in districts 10, 11, and 12 were considered with modified district marginals, in an attempt to illustrate the coincidences and differences between methods observed with the full data.} Table \ref{tab:votes} shows the votes disaggregated in all three dimensions.

\begin{table}[t]
\centering
\begin{tabular}{|c|r|r|r|l|c|r|r|r|}
\cline{1-4} \cline{6-9}
\textbf{Female} &
  \multicolumn{1}{c|}{\textbf{D1}} &
  \multicolumn{1}{c|}{\textbf{D2}} &
  \multicolumn{1}{c|}{\textbf{D3}} &
   &
  \textbf{Male} &
  \multicolumn{1}{c|}{\textbf{D1}} &
  \multicolumn{1}{c|}{\textbf{D2}} &
  \multicolumn{1}{c|}{\textbf{D3}} \\ \cline{1-4} \cline{6-9} 
\textbf{A} & 583494 & 365796 & 364104 &  & \textbf{A} & 242112 & 484302 & 145398 \\ \cline{1-4} \cline{6-9} 
\textbf{B} & 61674  & 48078  & 43416  &  & \textbf{B} & 20454  & 18762  & 23238  \\ \cline{1-4} \cline{6-9} 
\textbf{C} & 192546 & 431472 & 857646 &  & \textbf{C} & 231600 & 461640 & 452682 \\ \cline{1-4} \cline{6-9} 
\textbf{D} & 110664 & 227484 & 39774  &  & \textbf{D} & 65268  & 72696  & 36060  \\ \cline{1-4} \cline{6-9} 
\textbf{E} & 0       & 0       & 0       &  & \textbf{E} & 160686 & 0       & 0       \\ \cline{1-4} \cline{6-9} 
\textbf{F} & 0       & 0       & 74064  &  & \textbf{F} & 0       & 0       & 99618  \\ \cline{1-4} \cline{6-9} 
\end{tabular}
\caption{Votes disaggregated by districts, lists, and genders in our example used to illustrate the methods. Lists A to D have candidates of both genders in all districts, while lists E and F each compete in a single district.}
\label{tab:votes}
\end{table}

\subsection{The Chilean Constitutional Convention Method (CCM)}

We start by describing the method used in the recent Chilean Constitutional Convention election on May 15-16, 2021. In this method, the seats of each district $d\in D$ are distributed across the lists according to their votes using the ($1$-dimensional) Jefferson/D’Hondt method, i.e., we compute a vector
\[
    r^d \in \calA^1(\calV_{d,\cdot,+},q_d)
\]
for each $d\in D$. A second distribution of the $r^d_\ell$ seats assigned to each list in each district among its sublists is then performed through a new round of the Jefferson/D'Hondt method and the corresponding number of top-voted candidates of each sublist are provisionally selected; 
if at this point the elected candidates fulfill gender balance,\footnote{Gender balance in this context refers to the same number of candidates of each gender if the number of seats of the district is even and one more man/woman otherwise.} the seats are assigned. Otherwise, the elected candidate with the smallest number of votes among those of the over-represented gender is replaced by the non-elected candidate with the most votes among those of the other gender and the same sublist. This replacement procedure is repeated until gender balance is achieved. We refer to the Chilean electoral laws \citep{ley2015,ley2020} for the legal description of the method, and to \cite{mathieu2022apportionment} for an algorithmic and axiomatic analysis. 

\subsection{The $3$-proportional Method (TPM)}
\label{subsubsec:TPM}

The $3$-proportional method directly applies the notion of multidimensional proportionality introduced in Section \ref{sec:prelims} for the case $d=3$ with $N_1=D,\ N_2=L$ and $N_3=G$.
As described there, a $3$-proportional apportionment does not always exist, but an $\alpha$-approximate $3$-proportional apportionment is guaranteed to exist as long as $\sum_{i=1}^3 1/(\alpha_i+2)\le 1$. 
In particular, this holds when $\alpha_1=1$, $\alpha_2=0$ and $\alpha_3=4$. 

In order to properly define the method, we need to define the marginals. For the case of districts, we naturally set $m_d=q_d$ for each $d\in D$. For each list $\ell\in L$, we let $m_\ell$ be the number of seats corresponding to the total number of votes that the list gets in the country according to the Jefferson/D'Hondt method, i.e., $m_\ell = (\calA^1(\calV_{+,\cdot,+},H))_\ell$. For each gender we aim to assign half of the seats; when $H$ is odd we simply break the tie in favor of the gender obtaining the highest number of votes:
\[
    m_g = \begin{cases}
        \left\lceil \frac{H}{2} \right\rceil & \text{if } \calV_{+,+,g} > \frac{1}{2}\calV_{+,+,+},\\
        \left\lfloor \frac{H}{2} \right\rfloor & \text{otherwise.}
    \end{cases}
\]
We then aim to find a $3$-proportional apportionment with no deviation from the districts and list marginals and the least possible deviation from gender marginals, i.e., we set $\alpha_1=\alpha_2=0$ and
\[
    \alpha_3=\min\{z\in \{0,\ldots,4\}: \calA^3(\calV,m,H;(0,0,z)) \not= \emptyset\},
\]
and we naturally consider $x=\calA^3(\calV,m,H;\alpha)$. In case $\calA^3(\calV,m,H;(0,0,z)) = \emptyset$ for every \mbox{$z\in \{0,\ldots,4\}$}, we allow a deviation of $1$ with respect to district marginals, i.e., we fix $\alpha_1=1,~ \alpha_2=0$, and
\[
    \alpha_3=\min\{z\in \{0,\ldots,4\}: \calA^3(\calV,m,H;(1,0,z)) \not= \emptyset\},
\]
and we return $x=\calA^3(\calV,m,H;\alpha)$. The correctness of this procedure is guaranteed by the aforementioned existence result by \citet{cembrano2022multidimensional}. Once the seats $x_{d\ell g}$ for each tuple $(d,\ell,g)\in D\times L \times G$ are fixed,
they are distributed among the sublists, as before, through a new round of the ($1$-dimensional) Jefferson/D'Hondt method. The corresponding number of top-voted candidates of each district, sublist, and gender are then selected.

The result of the described method is shown in Table \ref{tab:result-TPM}.
\begin{table}[t]
\centering
\begin{tabular}{|c|c|c|c|c|c|c|c|c|}
\cline{1-4} \cline{6-9}
\textbf{Female} & \textbf{D1} & \textbf{D2} & \textbf{D3} &  & \textbf{Male} & \textbf{D1} & \textbf{D2} & \textbf{D3} \\ \cline{1-4} \cline{6-9} 
\textbf{A} & 3 & 2 & 2 &  & \textbf{A} & 1 & 4 & 1 \\ \cline{1-4} \cline{6-9} 
\textbf{B} & 1 & 0 & 0 &  & \textbf{B} & 0 & 0 & 0 \\ \cline{1-4} \cline{6-9} 
\textbf{C} & 0 & 2 & 5 &  & \textbf{C} & 1 & 4 & 3 \\ \cline{1-4} \cline{6-9} 
\textbf{D} & 0 & 2 & 0 &  & \textbf{D} & 0 & 1 & 0 \\ \cline{1-4} \cline{6-9} 
\textbf{E} & 0 & 0 & 0 &  & \textbf{E} & 0 & 0 & 0 \\ \cline{1-4} \cline{6-9} 
\textbf{F} & 0 & 0 & 0 &  & \textbf{F} & 0 & 0 & 1 \\ \cline{1-4} \cline{6-9} 
\end{tabular}
\caption{Number of seats assigned to candidates of each tuple given by a district, a list, and a gender in the $3$-proportional method described in Section \ref{subsubsec:TPM}.}
\label{tab:result-TPM}
\end{table}
This solution is actually found with deviations $\alpha_1 = \alpha_2 = \alpha_3 = 0$, so if we sum the apportionment across gender and lists we obtain the corresponding district magnitudes of 6, 15, and 12; the gender allocation respects the rule with 17 and 16 seats for women and men; and the list allocation is the one given by the Jefferson/D'Hondt method with a multiplier $\lambda = 6\cdot 10^{-6}$.

\subsection{The $3$-proportional Method with Threshold (TPM3)}
\label{subsubsec:TPM3}

In this method, we incorporate a threshold on the percentage of votes obtained by a list in order to be eligible for the apportionment on top of the $3$-proportional method.
More specifically, we only include in the process the set of lists $\ell\in L$ that obtain at least a 3\% of the votes. In order to implement this, we set $U_\ell=0$ for every $\ell\in L$ with $\calV_{+,\ell,+} < 0.03 \calV_{+,+,+}$, $U_\ell=H$ for every other list $\ell$, and $I_\ell=0$ for every $\ell\in L$. We then set $m_\ell=(\calA^1(\calV_{+,\cdot,+},H,I,U))_\ell$ and run the rest of the method in the same way as TPM.

\subsection{The $3$-proportional Method with Plurality (TPP)}
\label{subsubsec:TPP}

In this method, we modify the $3$-proportional method to ensure that the top-voted candidate of each district is elected.
The procedure is identical to TPM except that now we add $I,U\in \NN^{D\times L\times G}$ as arguments when applying $\calA^3$. Letting
\[ 
    c^*(d) = \arg\max_{c\in \calC: \dis(c)=d} v(c)
\]
denote the top-voted candidate of district $d$ for each $d\in D$, $I$ is defined as
\[
    I_{d\ell g} = \begin{cases}
        1 & \text{if } \lis(c^*(d))=\ell \text{ and } \gen(c^*(d)) = g,\\
        0 & \text{otherwise,}
    \end{cases} \qquad \text{for every } (d,\ell,g)\in D\times L\times G,
\]
and $U_{d\ell g} = H$ for every $(d,\ell,g)\in D\times L\times G$. This modification guarantees that $c^*(d)$ is elected for every $d\in D$.

\subsection{The $3$-proportional Method with Threshold and Plurality (TPP3)}\label{subsubsec:TPP3}

In this method, we ensure that the top-voted candidate of each district is elected as in TPP and we consider the threshold condition as in TPM3.
In order to implement this, we proceed in the same way as described in Section \ref{subsubsec:TPP}, except that we first compute the list marginals in order to respect the threshold condition, ensuring that each list below the threshold gets precisely its amount of top-voted candidates. Specifically, we define $I_{d\ell g}$ and $U_{d\ell g}$ as in Section \ref{subsubsec:TPP}, we set $U'_\ell=\sum_{d\in D}\sum_{g\in G} I_{d\ell g}$ for every $\ell\in L$ with $\calV_{+,\ell,+} < 0.03 \cdot \calV_{+,+,+}$, $U'_\ell=H$ for every other list $\ell$, and $I'_\ell=0$ for every $\ell\in L$. We then set $m_\ell=(\calA^1(\calV_{+,\cdot,+},H,I',U'))_\ell$ and run the rest of the method in the same way as TPP.

\paragraph{Comparison summary.} Table \ref{tab:methods-summary} summarizes the results for each method and Figure \ref{fig:apports} shows them in further detail, with the corresponding number of seats by list, district, and gender. The figure shows how the $3$-proportional methods achieve gender parity only across districts and lists, while the constraint is active in each district in the case of CCM. We can also observe how the threshold prevents list F from having a seat that was originally allocated by proportionality. List E does not have enough votes to obtain representation in TPM but gets a seat in both TPP and TPP3 due to having the most-voted candidate in district 1.\footnote{Observe that this happens despite the list getting less than 3\% of the votes as well, as the method fixes the number of seats for the lists with less than 3\% of the votes to their number of top-voted candidates.}
\begin{table}[t]
\centering
\begin{tabular}{c|r|r|r|r|r|r|}
\cline{2-6}
\textbf{} &
  \multicolumn{1}{c|}{\textbf{CCM}} &
  \multicolumn{1}{c|}{\textbf{TPM}} &
  \multicolumn{1}{c|}{\textbf{TPM3}} &
  \multicolumn{1}{c|}{\textbf{TPP}} &
  \multicolumn{1}{c|}{\textbf{TPP3}} \\ \hline
\multicolumn{1}{|c|}{\textbf{A}} & 13  & 13 & 13 & 12 & 13 \\ \hline
\multicolumn{1}{|c|}{\textbf{B}} & 0   & 1  & 1  & 1  & 1  \\ \hline
\multicolumn{1}{|c|}{\textbf{C}} & 17  & 15 & 16 & 15 & 15 \\ \hline
\multicolumn{1}{|c|}{\textbf{D}} & 2   & 3  & 3  & 3  & 3  \\ \hline
\multicolumn{1}{|c|}{\textbf{E}} & 0   & 0  & 0  & 1  & 1  \\ \hline
\multicolumn{1}{|c|}{\textbf{F}} & 1   & 1  & 0  & 1  & 0  \\ \hline
\end{tabular}
\caption{Comparison of seat allocation by list under each method.}
\label{tab:methods-summary}
\end{table}
\begin{figure}[h!]
    \centering
    \includegraphics[width=\textwidth]{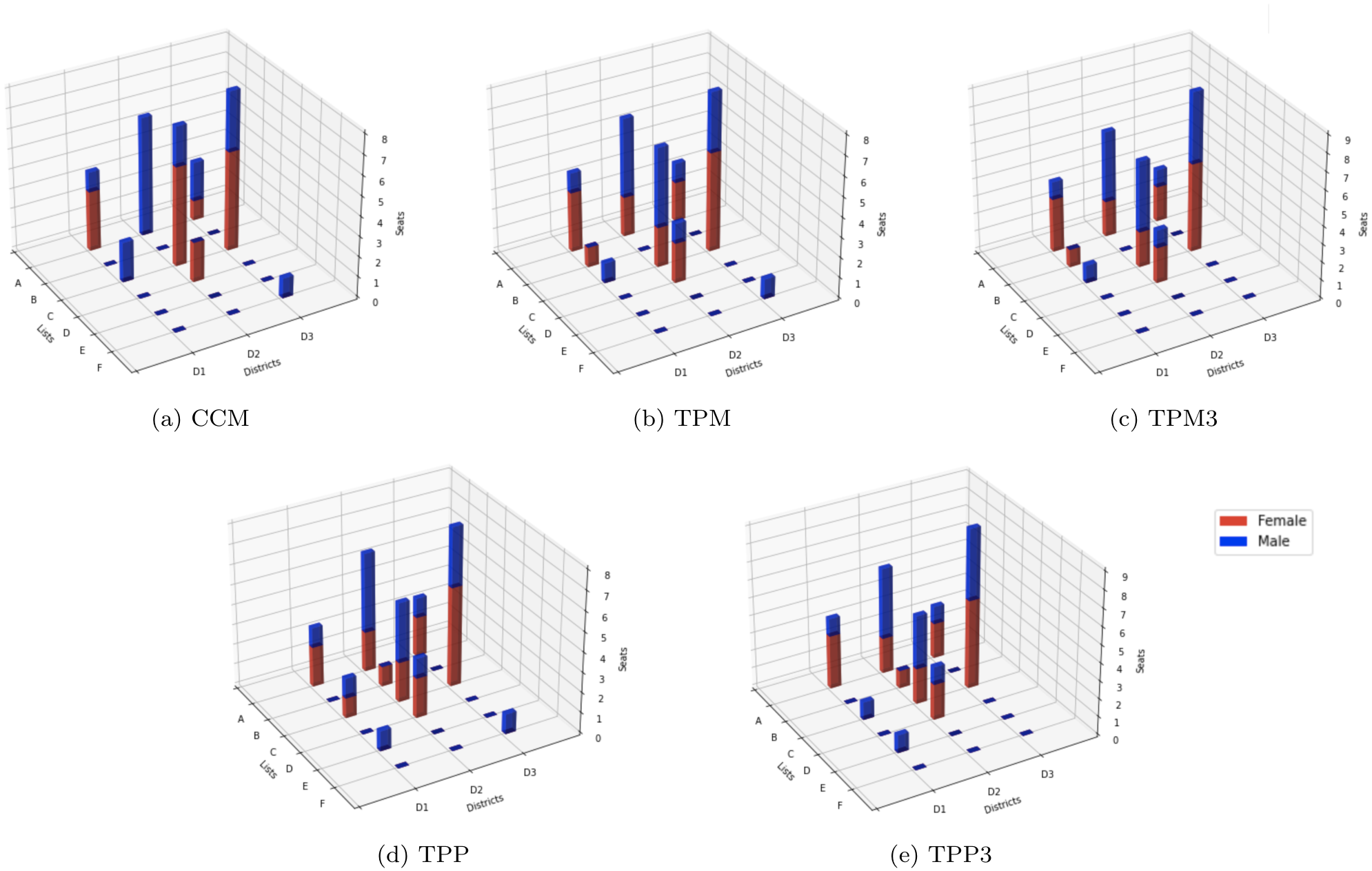}
    \caption{Comparison of the apportionment of seats in each method, by lists, districts, and gender. Each bar height represents the number of seats obtained by a list in a certain district, segmented by the number of seats corresponding to men (blue) and women (red).}
    \label{fig:apports}
\end{figure}

\section{Experimental Results and Analysis}\label{sec:results}

In this section, we study the outcomes of the methods described above with the data from the 2021 Chilean Constitutional Convention election. We evaluate their performance in terms of the proportionality of the results, the representativeness of the elected candidates, and the value of each cast vote. We also use the election data to test the deviations from the marginals required by the methods based on multidimensional proportionality.

We first give some context of the Chilean Constitutional Convention election.  While $155$ seats were to be allocated in total, $17$ of these seats were reserved for ethnic minority groups and assigned through a parallel election; our analysis thus focuses on the $138$ non-ethnic seats. The country is divided into 28 electoral districts, each receiving between $3$ and $8$ seats. Candidates run in open lists and each voter votes for at most one candidate.

A total of $70$ lists and over $1300$ candidates competed in the election. Among them, only $20$ lists and two independent candidates\footnote{These candidates are denoted as IND1 and IND9 (according to the districts in which they run).} obtained enough votes to be elected in any of the methods described in Section \ref{sec:methods}. For exposition purposes, in Sections \ref{subsec:prop} and \ref{subsec:representativeness} we omit the other lists and independent candidates, which does not affect our results.\footnote{None of these lists and independent candidates obtained more than $0.51\%$ of the votes, and they jointly represent less than the $10\%$ of the total votes.} Among the main political alliances, we encounter the XP list, which encompasses traditional and newer right-wing parties, the YB list, formed by the center-left parties that governed Chile between 1990 and 2010, and the YQ list, which contains left-wing parties. In addition, two conglomerates of lists containing politically independent candidates obtained considerable support; following their campaign we denote them by LP (for \emph{Lista del Pueblo}) and INN (for \emph{Independientes No Neutrales}).\footnote{This grouping is standard; see, e.g., \url{https://2021.decidechile.cl/\#/ev/2021}. 
Full election data is available on the website of \citet{servel2021data}.}
Besides these two lists, we use the election codes to represent the lists.

We remark that a deviation of one seat from the gender marginals was needed when running TPM3 and TPP3 on the electoral data, while no deviation was needed in the case of TPM and TPP. In Section \ref{subsec:deviations}, we will see that this is no coincidence, as deviations needed in practice tend to be much smaller than the worst-case guarantees.

\subsection{Proportionality}\label{subsec:prop}

Proportionality in the election results is measured through the difference between the induced political distribution and the ($d$-dimensional) fair share, which was introduced in Section \ref{sec:prelims} and constitutes a natural notion of perfectly fair distribution. Figure \ref{fig:parliaments} provides a graphical view of the political representation obtained by the lists under each method and this perfectly fair political distribution. The data used for this figure is contained in Table \ref{tab:votes-bylist} in Appendix \ref{app:tables}.

\begin{figure*}[t]
     \centering
     \begin{subfigure}[b]{0.3\textwidth}
         \centering
         \includegraphics[scale=.13]{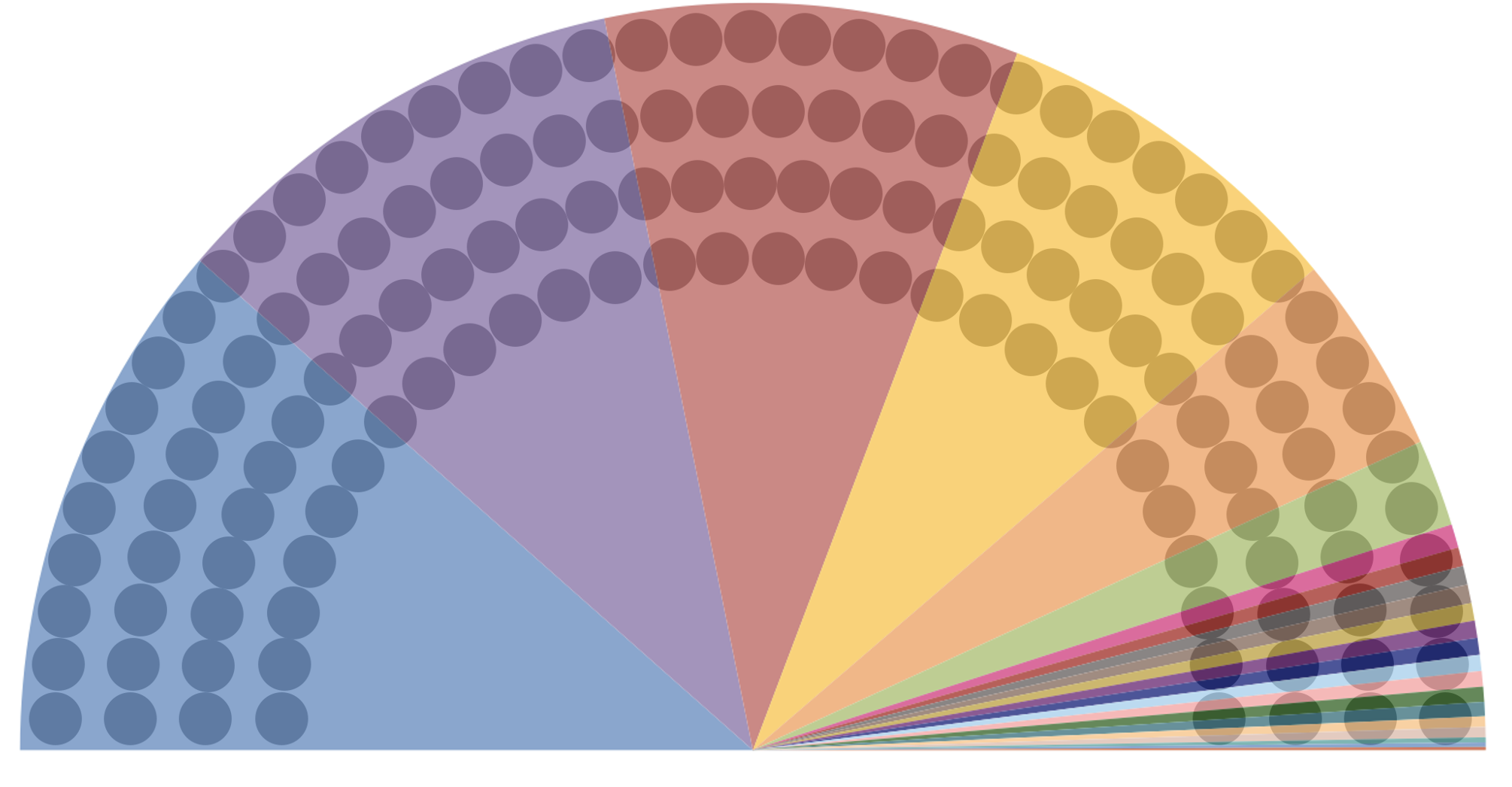}
         \caption{Fair share}
     \end{subfigure}
     \hfill
     \begin{subfigure}[b]{0.3\textwidth}
         \centering
         \includegraphics[scale=.13]{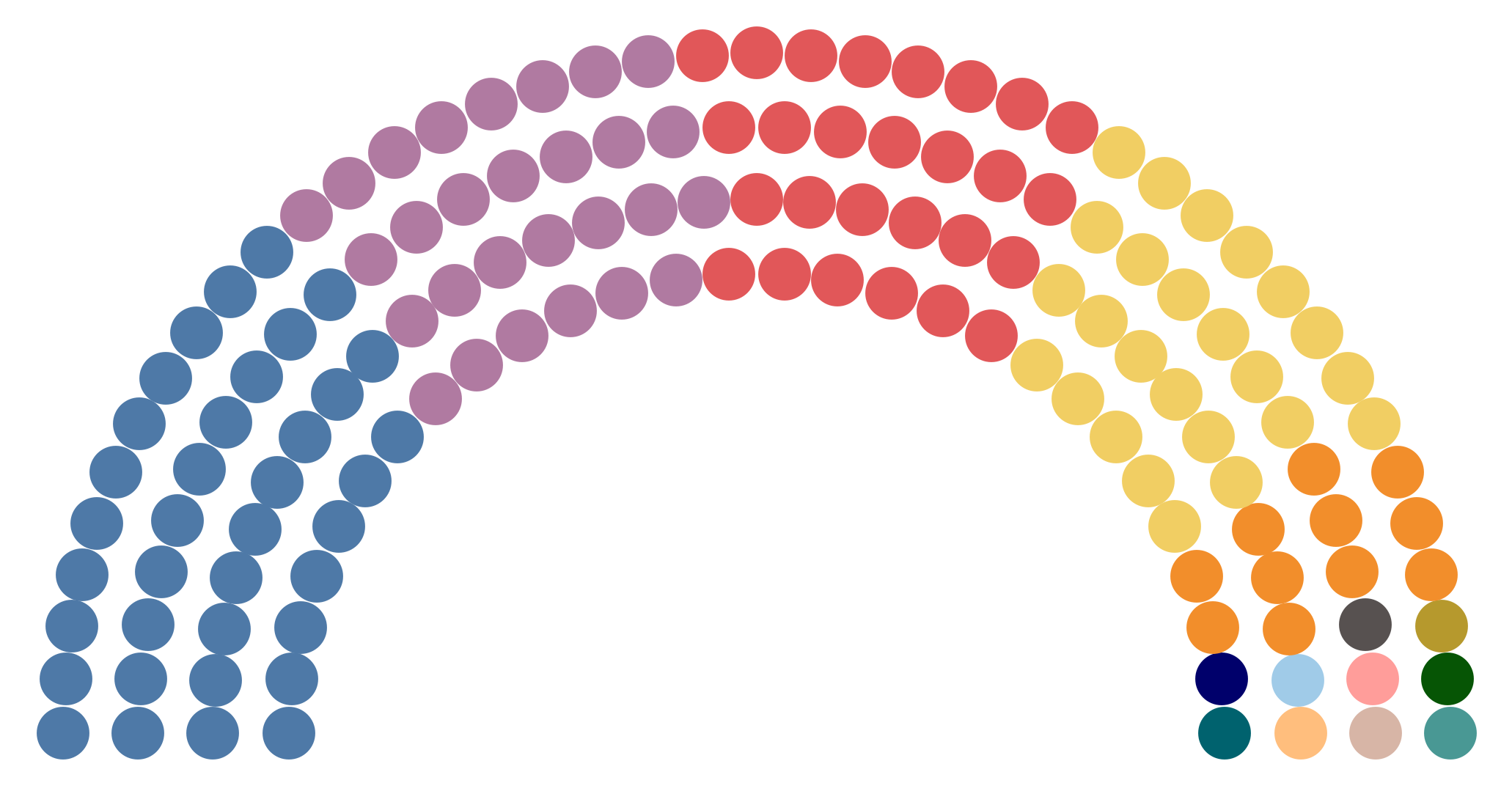}
         \caption{CCM}
     \end{subfigure}
     \hfill
     \begin{subfigure}[b]{0.3\textwidth}
         \centering
         \includegraphics[scale=.13]{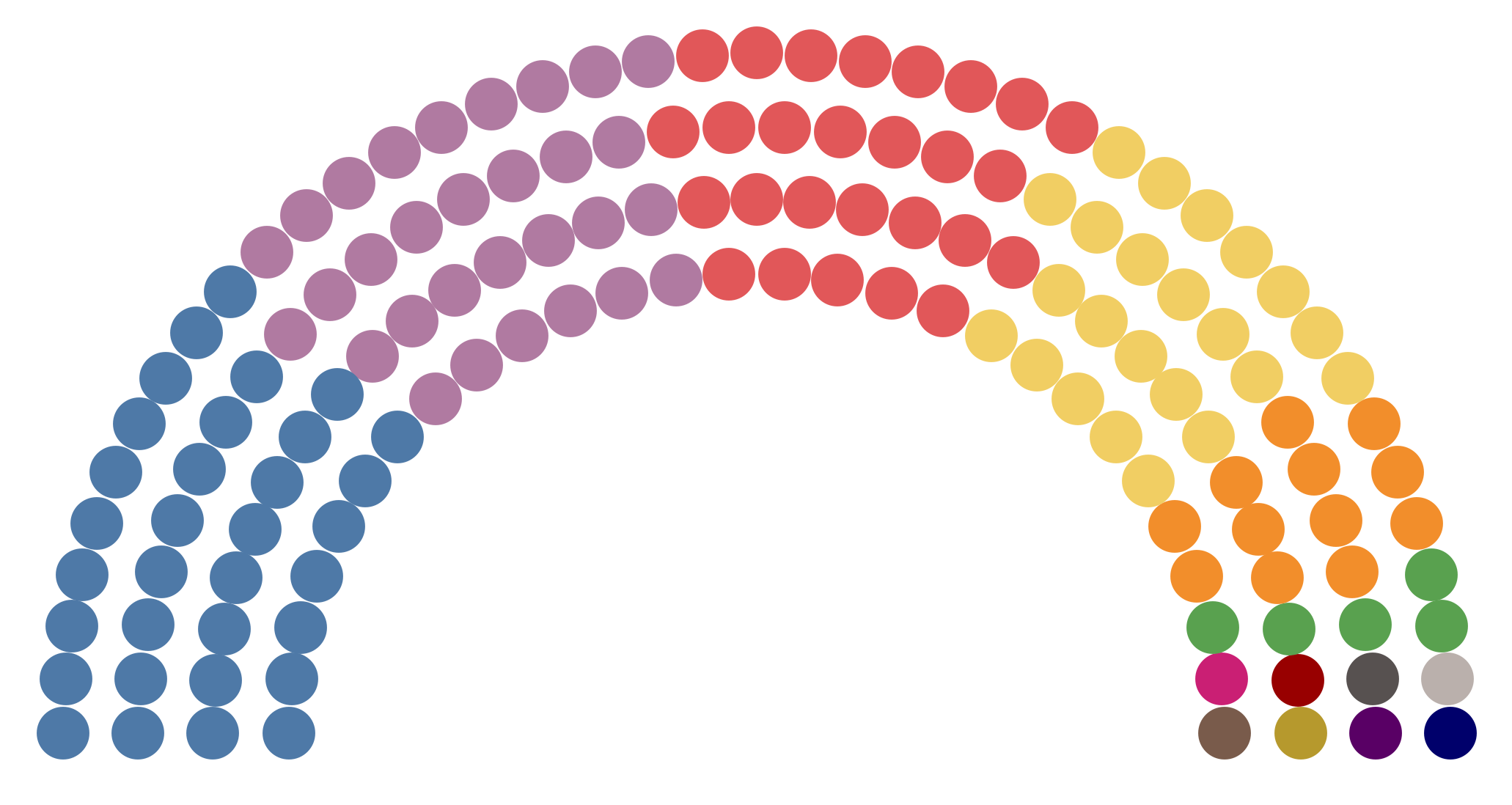}
         \caption{TPM}
     \end{subfigure}
     \hfill
     \begin{subfigure}[b]{0.3\textwidth}
         \centering
         \includegraphics[scale=.13]{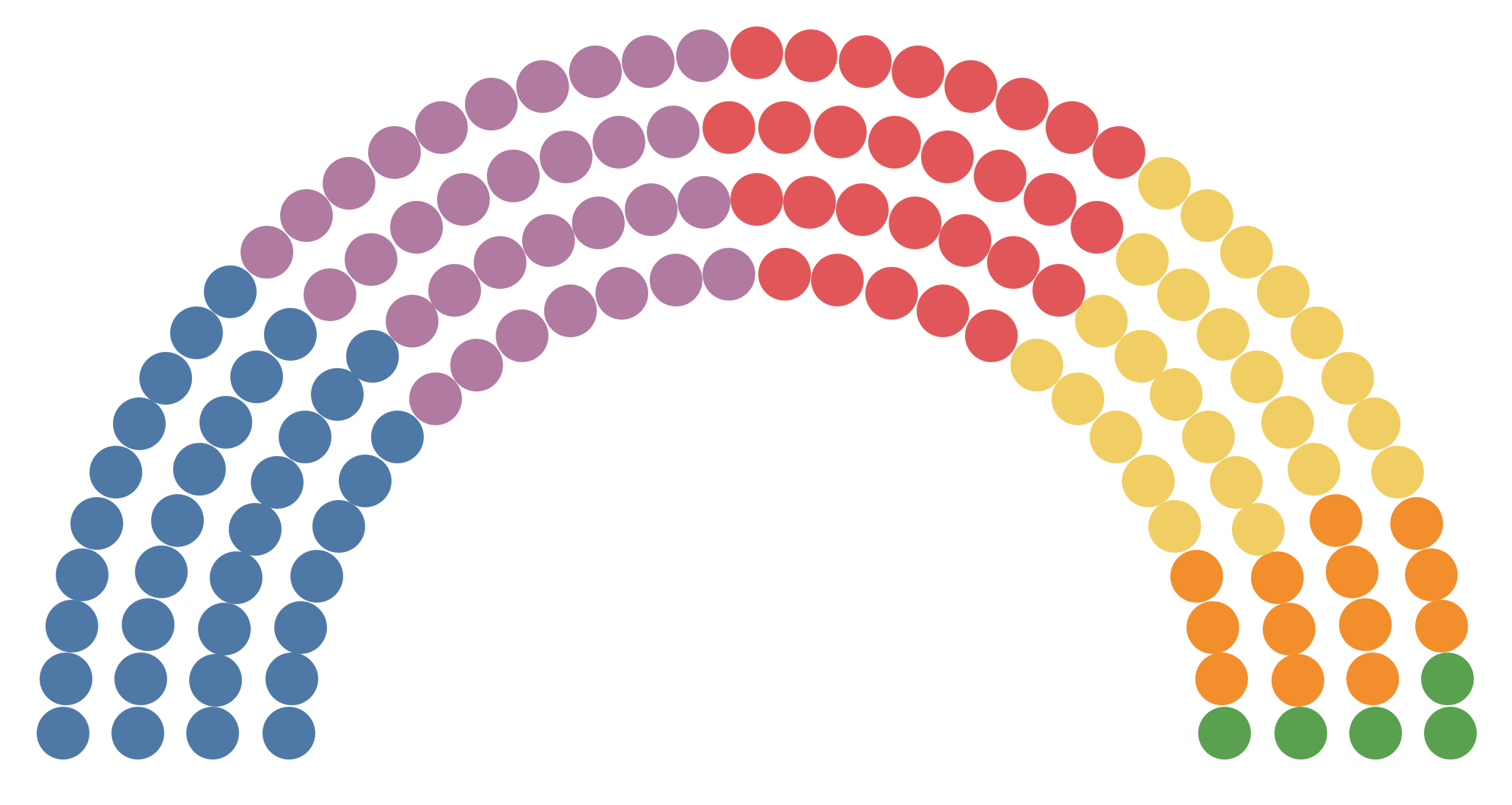}
         \caption{TPM3}
     \end{subfigure}
     \hfill
     \begin{subfigure}[b]{0.3\textwidth}
         \centering
         \includegraphics[scale=.13]{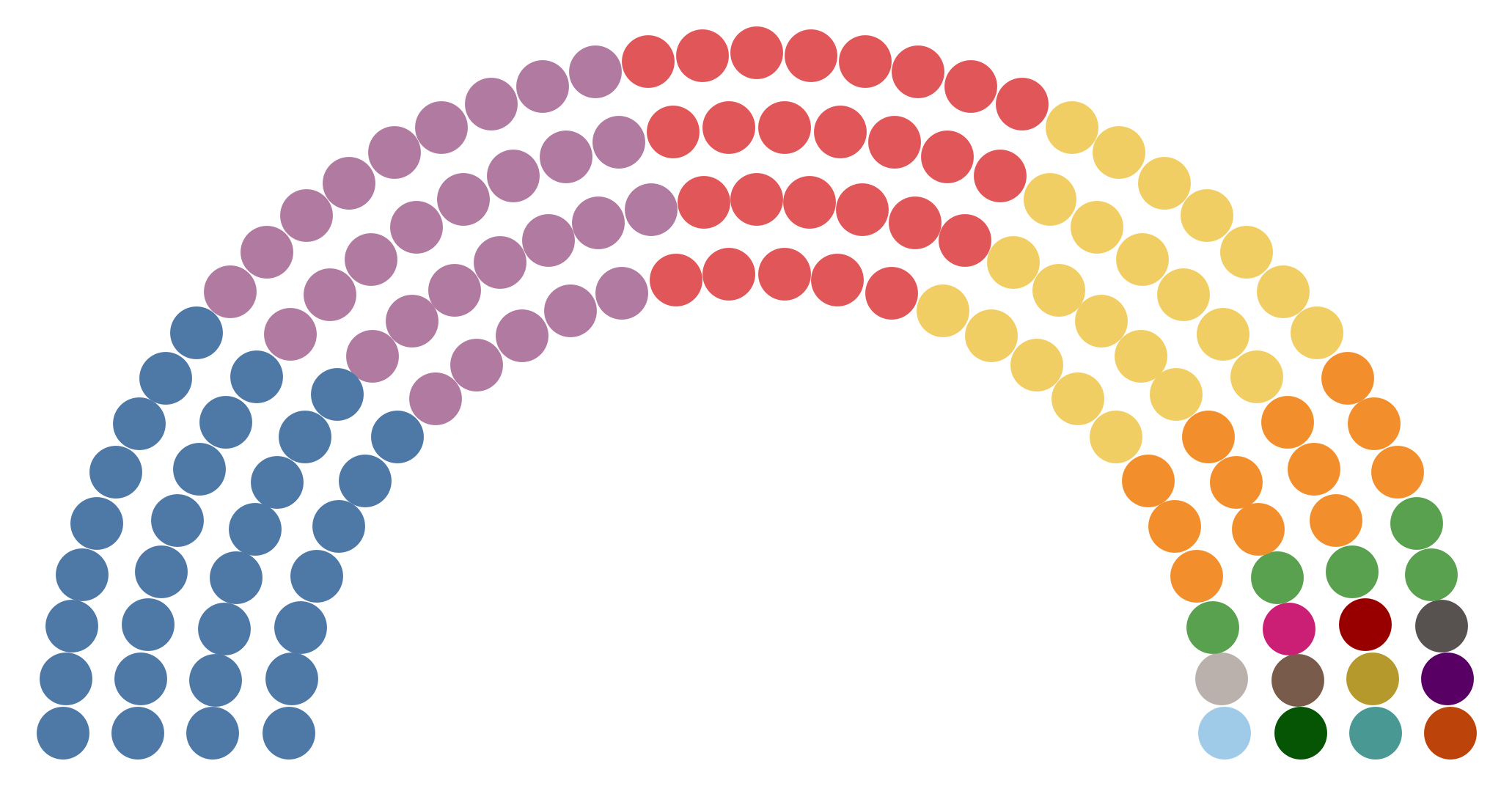}
         \caption{TPP}
     \end{subfigure}
     \hfill
     \begin{subfigure}[b]{0.3\textwidth}
         \centering
         \includegraphics[scale=.13]{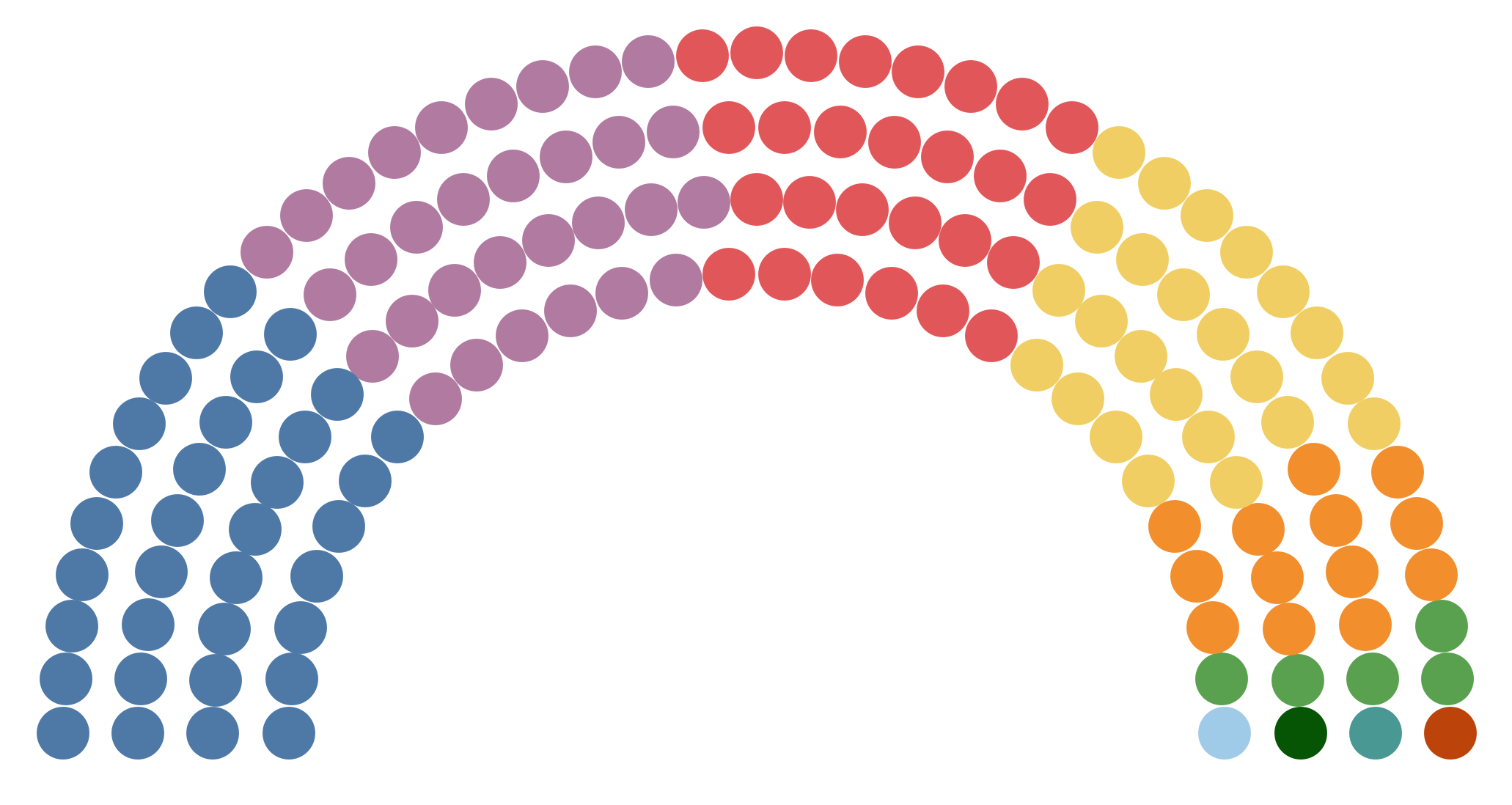}
         \caption{TPP3}
     \end{subfigure}
     \begin{subfigure}[b]{\textwidth}
         \centering
         \includegraphics[scale=.5]{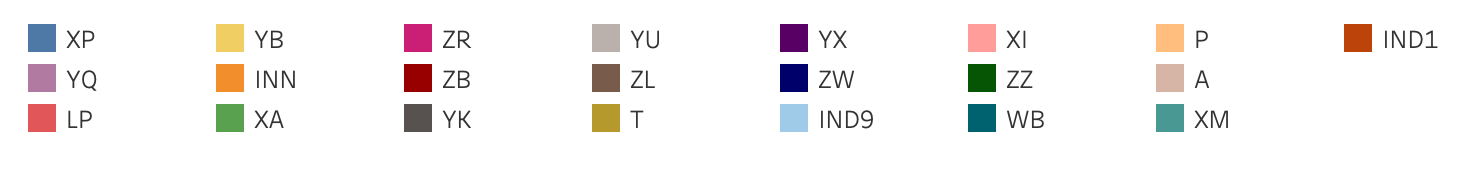}
     \end{subfigure}
        \caption{Political representation obtained by list under each method and fair share apportionment for comparison.}
\label{fig:parliaments}
\end{figure*}

It is observed that the global methods, especially those without a threshold (TPM and TPP), produce a political distribution much closer to the fair share than the local method CCM. In particular, the most-voted list XP is not as overrepresented as in CCM, and seats are assigned to a subset of the top-voted lists as close to proportionality as possible.
As an illustrative example, one may observe that the list XA---the sixth most-voted one with almost 4\% of the votes---receives five seats under the global methods and no seat under CCM. The result under TPM and TPP is a highly varied convention made up of large political blocks together with multiple lists that obtain a single representative.
While TPM3 considerably reduces the number of lists entering the convention, the allocation produced by TPP3 stands as a midpoint between those produced by TPM3 and TPP, leaving out lists with too few votes at a national level but allowing the representation of strong local projects. In particular, the XM list, which is a project of the southernmost part of Chile (Magallanes), as well as the independents from districts 1 and 9, which were the first majorities of their districts, enter the convention.

The deviation from the fair share can be naturally quantified through the Euclidean distance between a specific apportionment and the fair share, a value called \textit{Gallagher Index} in the literature. This deviation can be measured either globally in the country or locally within each district (taking the average thereof). We call the respective values \textit{global Gallagher Index} and \textit{local Gallagher Index}. Furthermore, we can define a $3$-dimensional Gallagher Index in order to evaluate the proportional allocation of seats across districts, lists, and genders simultaneously: We take the Euclidean distance between the apportionment and the $3$-dimensional fair share defined in Section \ref{sec:prelims}.

Table \ref{tab:props} shows the described indices for each method.\footnote{The Gallagher Index computed by district can be found in Table \ref{tab:GI-by-dist} in Appendix \ref{app:tables}.}
\begin{table}[t]
\centering
\begin{tabular}{r|r|r|r|r|r|}
\cline{2-6}
\multicolumn{1}{l|}{} &
  \multicolumn{1}{c|}{\textbf{CCM}} &
  \multicolumn{1}{c|}{\textbf{TPM}} &
  \multicolumn{1}{c|}{\textbf{TPM3}} &
  \multicolumn{1}{c|}{\textbf{TPP}} &
  \multicolumn{1}{c|}{\textbf{TPP3}} \\ \hline
\multicolumn{1}{|r|}{Global GI} & 4.6  & 1.8  & 3.8  & 1.3 & 2.9  \\ \hline
\multicolumn{1}{|r|}{Local GI avg.}  & 14.5 & 18.5 & 19.0 & 18.2 & 18.5 \\ \hline
\multicolumn{1}{|r|}{$3$-dim GI}  & 3.9  & 3.6  & 3.8  & 3.6  & 3.8  \\ \hline
\end{tabular}
\caption{Gallagher Index (\%) computed for the global political distribution, the average of the local political distributions, and the $3$-dimensional seat distribution.}
\label{tab:props}
\end{table}
Naturally, the district averages show that the local method (CCM) is closer to the local fair share than global methods, as the former is designed to pursue local proportionality. However, the local errors produced by CCM add up across districts, leading to a considerably larger global Gallagher Index than global methods---particularly than TPM and TPP---and a slightly larger $3$-dimensional index. This is again natural, as TPM and TPP compute the seats to be allocated to each list in a global manner, thus leading to a political distribution that is closer to proportionality on a national level. In fact, TPP ends up with the lowest global deviation followed by the TPM, while TPP3 again stands as a midpoint between the methods with and without a threshold in this regard.
In terms of the $3$-dimensional index, the results are similar between all the methods, with differences of at most 0.3 percentage points. This similarity can be explained by the fact that the $3$-proportional methods generate a good allocation to lists but produce local distortions.
Conversely, the national allocation to lists does not adjust correctly to the votes in CCM but the district allocation does, highlighting a fundamental trade-off between local and global political representation.

\subsection{Representativeness}
\label{subsec:representativeness}

In this section, we analyze the average of votes and percentages obtained by the elected candidates under each method. These values are summarized in Table \ref{tab:average-votes}. We additionally include, as a reference for the maximum attainable representativeness given the constraints of the election, an apportionment method that chooses the most-voted candidates constrained to district marginals and global gender parity. We denote this method by \textit{Greedy}.

\begin{table}[t]
\centering
\begin{tabular}{r|r|r|r|r|r|r|}
\cline{2-7}
\multicolumn{1}{l|}{} &
  \multicolumn{1}{c|}{\textbf{CCM}} &
  \multicolumn{1}{c|}{\textbf{TPM}} &
  \multicolumn{1}{c|}{\textbf{TPM3}} &
  \multicolumn{1}{c|}{\textbf{TPP}} &
  \multicolumn{1}{c|}{\textbf{TPP3}} &
  \multicolumn{1}{c|}{\textbf{Greedy}}\\ \hline
\multicolumn{1}{|r|}{Avg. votes} & 14126 & 14048 & 13848 & 14406 & 14198 & 15227  \\ \hline
\multicolumn{1}{|r|}{Avg. district \%} & 31.2\% & 31.0\% & 30.7\% & 32.1\% & 31.7\%  & 37.6\% \\ \hline
\end{tabular}
\caption{Average votes obtained by the elected candidates under each method, and the average percentage of votes obtained by elected candidates with respect to votes cast in their districts.}
\label{tab:average-votes}
\end{table}

By construction, all methods that involve correction mechanisms in order to ensure gender parity imply a certain degree of loss of votes, since candidates who may have been elected without corrections are substituted by other non-elected ones with a lower number of votes. 
A relevant observation is the fact that TPP and TPP3 achieve a higher average than CCM, despite the greedy nature of the latter: It just replaces candidates when necessary and in a local manner. One explanation is the presence of locally top-voted candidates who are not elected in other methods and obtain a large number of votes. This is particularly relevant since, in addition to the property of plurality---and a representation threshold in the case of TPP3---these methods obtain the best results in terms of representation, followed by CCM and TPM.
It is observed that the threshold decreases the average votes in this case. This behavior, however, is not a direct consequence of its application but rather depends on the instance. In this particular case, the negative effect on candidates of small lists who were not elected due to this threshold was more important than the positive effect on candidates of bigger lists. 

\subsection{One Person, One Vote?} 
\label{subsec:value-vote}

In this section, we compare the value of a cast vote under each method in terms of political impact.\footnote{Unlike the other subsections, in this one the votes of the lists without enough votes to obtain a seat under some method are considered, as they are relevant for the measure.} Under local methods (CCM), each vote counts only for electing the seats assigned to the corresponding district, and therefore its power can be measured by the ratio between the number of seats assigned to the district and the number of votes cast in the district. Figure \ref{fig:vote-value} shows this indicator for each district, divided by the overall ratio between the house size and votes to normalize. We observe that the votes of people living in central districts are less powerful, as defined previously, than the votes of people living in extreme districts, reaching a factor of up to 5.93 between them. Conversely, global methods without plurality ensure that, in terms of political representation, every vote is equally valuable as seats are assigned to political lists proportionally to the number of votes they receive. Global methods with plurality combine both features---global representation and election of locally top-voted candidates---so each vote is valuable in terms of national political distribution but has special relevance for the district where it is cast.

\begin{figure*}[t]
	\centering
	\includegraphics[scale=.65]{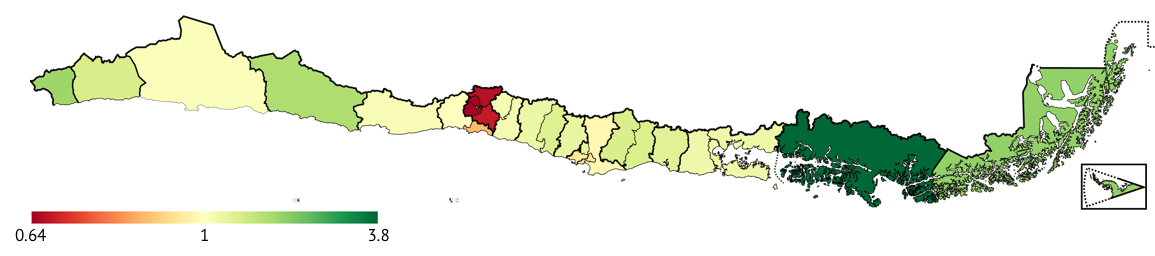}
	\caption{Vote power by district as ratio of global vote power.}
	\label{fig:vote-value}
\end{figure*}

This is related to the aforementioned trade-off between local and global representation. Incorporating the different criteria as additional dimensions, rather than running separate elections, may be a reasonable way to get closer to the widely-known principle of {\it one person, one vote}.

\subsection{Deviations in the Number of seats}
\label{subsec:deviations}

We finally conduct simulations to assess the deviations from the marginals that the proposed methods require in practice. Our simulations consider four vote models with different distributions and the same instance size as the original election (28 districts, 22 lists, 2 genders, and a house size of 138 seats). Denoting by $\mu$ the mean of the votes $\calV_{d\ell g}$ and $\calV_{d\ell g}^n$ the $n$-th random sample from the distribution $F$, we consider the following models:

\begin{enumerate}
    \item $\calV_{d\ell g}^n = \calV_{d\ell g}\cdot X_n$, where X is drawn from $\mathrm{Normal}(1, 0.1)$, and $\calV_{d\ell g}$ are the original votes;
    \item $\calV_{d\ell g}^n$ is drawn from $\mathrm{Poisson}(\theta_{d \ell})$, and $\theta_{d \ell} \sim \mathrm{Gamma}(1, \mu)$;
    \item $\calV_{d\ell g}^n$ is drawn from $\mathrm{Uniform}(1, 2\mu)$;
    \item $\calV_{d\ell g}^n$ is drawn from $\mathrm{Pareto}(0.5, 1000)$.
\end{enumerate}

The first instance aims to use a normal perturbation of the original votes. The second distribution emulates an election with heterogeneous preferences for lists across different districts, governed by the parameters $\theta_{d\ell}$. In the third and fourth models, we consider equally distributed entries, sampled from a uniform or a heavy-tailed Pareto distribution, respectively. Note that the number of votes and the number of seats per district may be too inconsistent in the last three models; to correct this, we scale each entry $(d,\ell,g)$ of each generated matrix by $\frac{P_d}{P_+} \cdot \frac{\calV_{+++}}{\calV_{d++}}$, where $P_d$ is the population of district $d$ and $P_+$ the total population (according to the last official census in 2017). We set the simulations to aim for a similar mean to the original vote matrix ($\mu \approx 4200$). Model 4 does not have a finite mean, but the median is similar.

In Figure \ref{fig:deviations} we can observe that, despite the upper bound of $\alpha=(1,0,4)$ on the deviations needed to find an approximate $3$-proportional apportionment, no deviation is needed in most simulations to find an output of our methods. When this is not the case, no deviation from the district marginals and a deviation of only one seat---or very rarely, two or three seats---from the gender marginals is required. 
The vote instances drawn from the uniform and Pareto distributions tend to perform slightly worse in this sense, which can be understood as a difficulty in finding multidimensional proportional apportionments when the preferences for lists and genders are distributed similarly overall but inconsistently across districts.

\begin{figure}[t]
    \centering
    \begin{subfigure}[b]{0.49\textwidth}
        \centering
        \includegraphics[width=\linewidth]{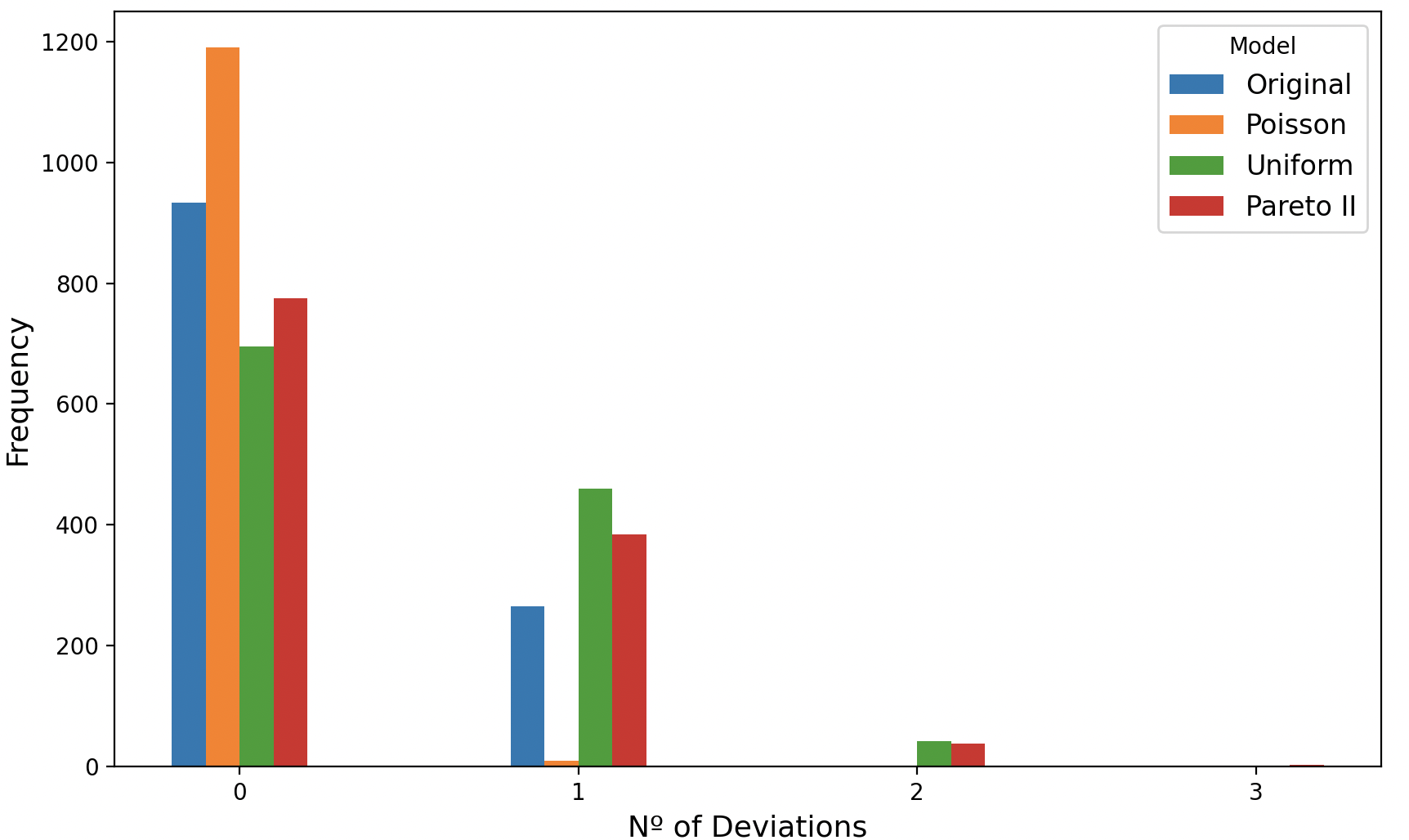}
        \caption{TPM}
    \end{subfigure}
    \hfill
    \begin{subfigure}[b]{0.49\textwidth}
        \centering
        \includegraphics[width=\linewidth]{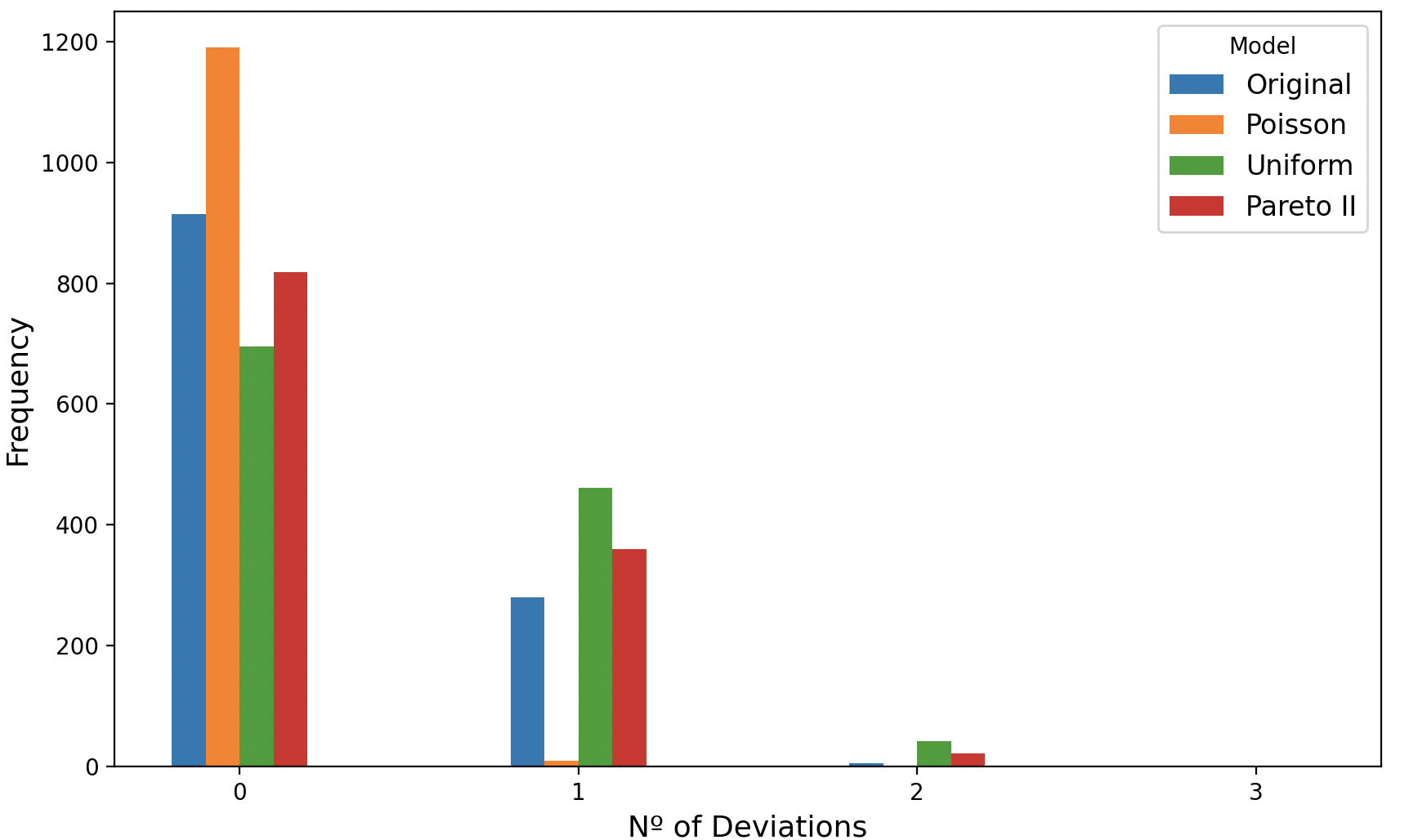}
        \caption{TPM3}
    \end{subfigure}
    \begin{subfigure}[b]{0.49\textwidth}
        \centering
        \includegraphics[width=\linewidth]{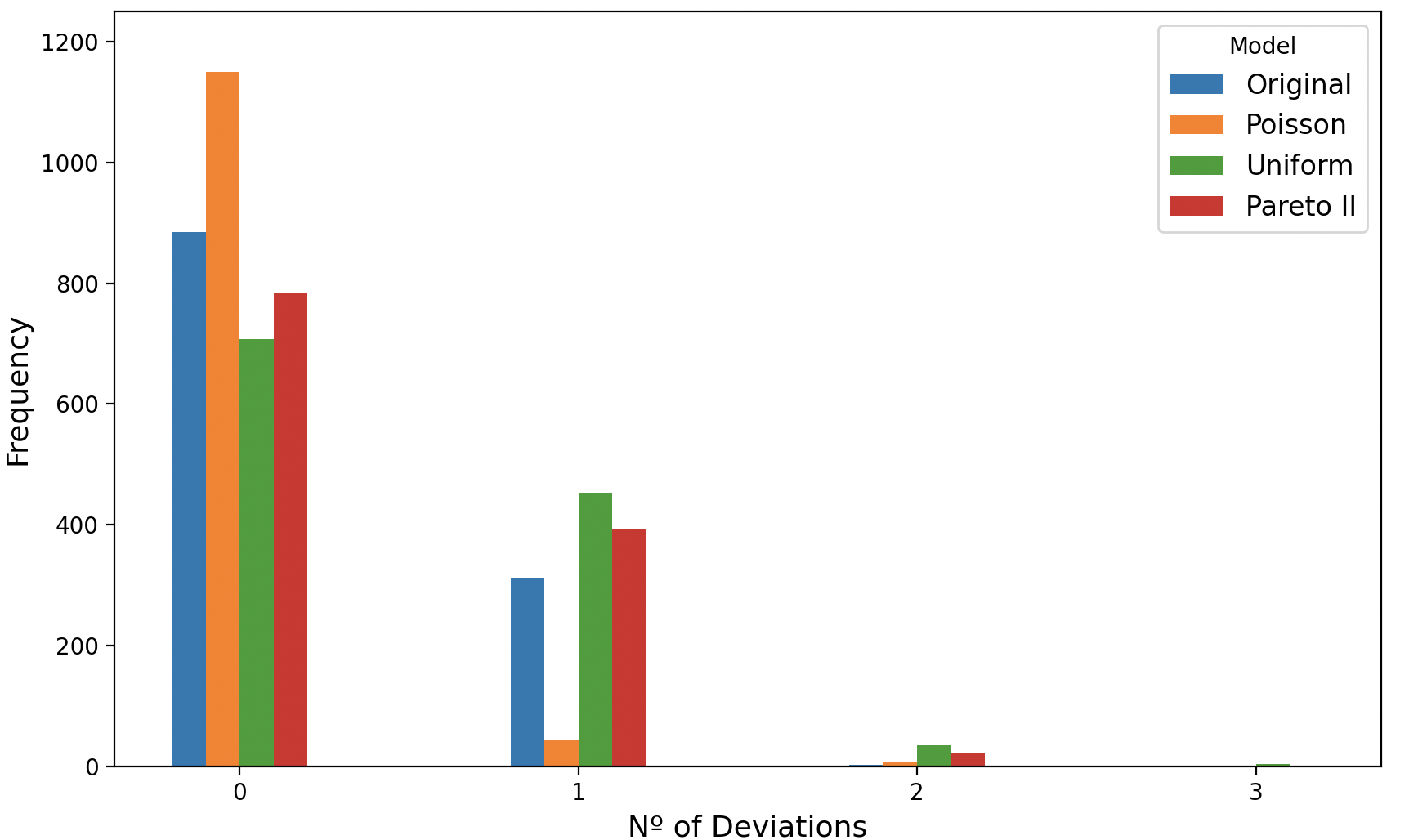}
        \caption{TPP}
    \end{subfigure}
    \hfill
    \begin{subfigure}[b]{0.49\textwidth}
        \centering
        \includegraphics[width=\linewidth]{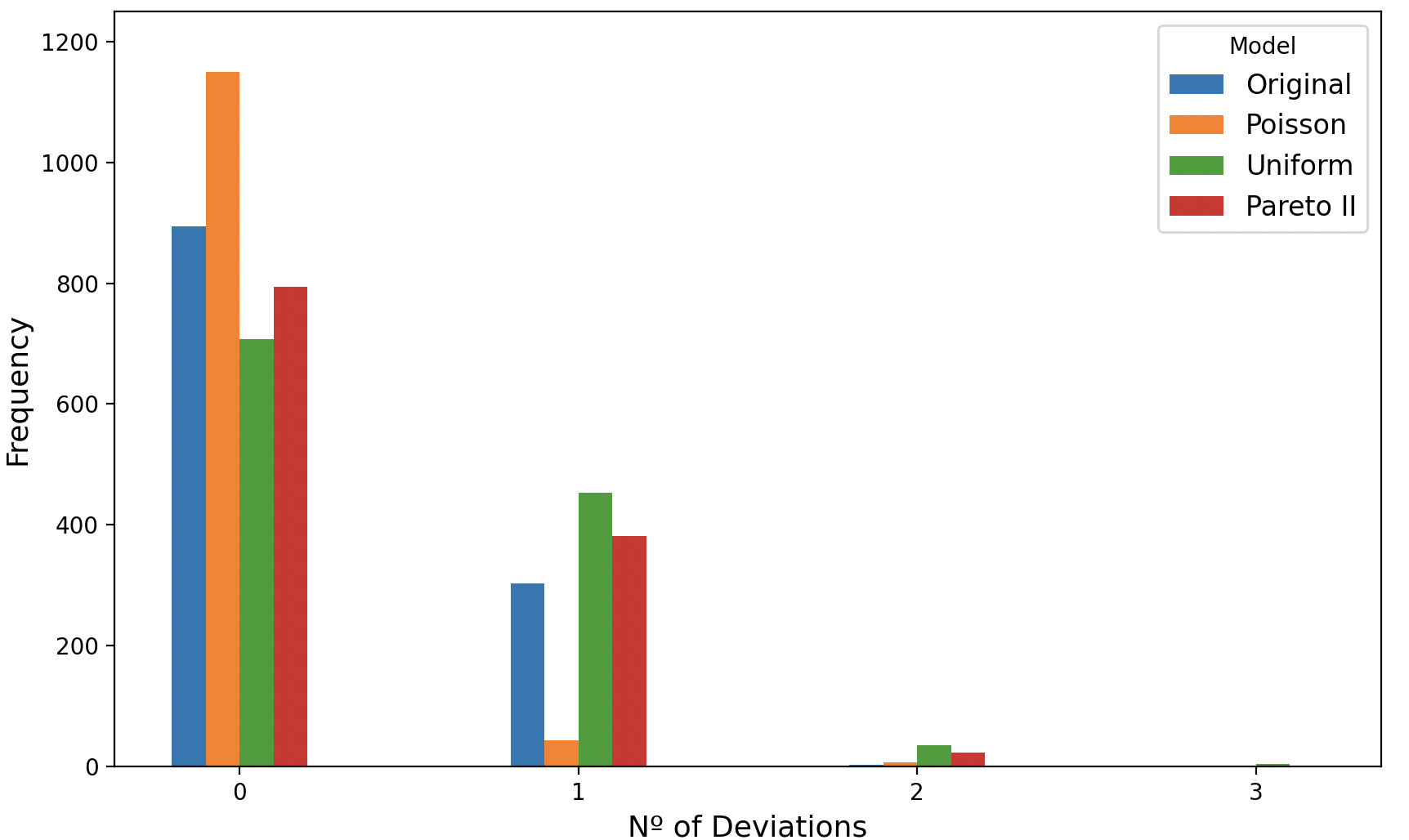}
        \caption{TPP3}
    \end{subfigure}
    \caption{Distribution of absolute deviations magnitude in gender marginals for the $3$-proportional methods. No deviations in other marginals were ever required. We conduct a simulation of $N=1200$ random vote instances from 4 different distributions and compute the frequency of the deviations for each method and instance.}
    \label{fig:deviations}
\end{figure}

\section{Discussion}

In this work, we have studied apportionment methods based on multidimensional proportionality from an applied perspective, establishing practically relevant theoretical results and proposing concrete methods capable of combining this notion with other desirable features of apportionment systems.

As our theoretical contribution, we have established two results that support the usage of multidimensional proportional apportionment methods as proposed by \citet{cembrano2022multidimensional}. We first showed that, in the fractional case, this notion is the only one satisfying three widely studied axioms, thereby extending a result for two dimensions \citep{balinski1989axiomatic} and justifying the corresponding notion of integral proportionality. We have then proved that the addition of bounds on the entries of the apportionment tensor does not affect the existence of approximately proportional apportionments, which is particularly relevant for adding plurality constraints on top of the basic method.

Regarding experimental results, we have used the Chilean Constitutional Convention election as a testing ground to analyze the outcomes of four methods based on multidimensional proportionality in terms of political representation, representativeness, and how well they accomplish the ideal of \textit{one person, one vote}. These methods consist of the basic method proposed by \citet{cembrano2022multidimensional} and three methods incorporating one or two additional features: a threshold on the number of votes that a list requires to receive any seats and the election of the most-voted candidate of each district.

Regarding proportionality in political representation, our results highlight a natural trade-off between global and local representation: Local methods (such as the one used in the election) generate global distortion in list allocation, while global methods based on multidimensional proportionality generate district distortion. In this trade-off, the incorporation of plurality constraints presents a middle ground, as highly-voted candidates who run as part of local projects are not excluded due to obtaining too few votes on a national scale. When considering the average vote of elected candidates, this method exhibits the best performance, while the threshold slightly lowers the average votes. Regarding voting power among districts, we observe that the ratio between votes and seats varies considerably across districts. This effect is attenuated in methods based on multidimensional proportionality, as the allocation of seats to lists is mainly given by the national amount of votes a list gets. Overall, both the threshold and the plurality constraints achieve their respective goals of favoring larger political coalitions and incorporating some degree of local representation. The cost of the former, in terms of proportionality and representativeness, seems to be larger than that of the latter in our simulations, though the exact way in which they are applied may of course have different impacts.

The presented methods can be valuable and feasible proposals for the election of representative bodies. Proportional or mixed-member proportional representation systems are two families of electoral systems widely used worldwide, and some countries have also started to use biproportional apportionment methods \citep{maier2010divisor,pukelsheim2017proportional}. Our theoretical and applied results lay the groundwork for incorporating features such as representation threshold---to favor the conformation of broader political projects---or the election of the most-voted candidate in every district---to enhance local representation---into methods based on multidimensional proportionality. In particular, a method incorporating both features such as TPP3 may constitute an appealing alternative to balance local and global proportionality, representativeness, and voting power.

Although multidimensional proportional methods may appear harder to understand for citizens at first glance, all methods incorporating additional constraints such as gender balance suffer from this problem, and the concepts of proportionality, gender balance, and plurality directly imposed in the proposed methods are certainly intuitive in contrast with the somehow subtle corrections made in the actual method used in the election of the Chilean Constitutional Convention. It is also important to mention that even though the ease of understanding constitutes a relevant element in terms of the legitimacy of the electoral process, a proper description and an appropriate performance of the mechanism do as well.
Moreover, $3$-proportional methods are easily extended to more dimensions to include, for instance, ethnicity as a fourth dimension and thus allow ethnic seats to be proportionally assigned across political, geographical, and gender distribution.
We remark that methods based on multidimensional proportionality are computationally efficient, both theoretically and in practice. Furthermore, the observed deviations from the prescribed marginals required to implement these methods are in most cases zero or very small.

All these facts consolidate the possibility of thinking deeper about electoral methods incorporating several dimensions and constraints. This future work certainly requires an interdisciplinary approach in order to succeed in designing mechanisms able to better represent the complexity and diversity of modern societies. 
For example, this work does not account for the effects of the new rules on the voters' and parties' strategic behavior. Another direction concerns the relationship of the electoral rules with the legitimacy of the elected representative bodies.
From a mathematical and computational perspective, this work shows that optimization and algorithmic tools are valuable for designing and testing new electoral methods for increasingly complex societies, and motivates the development of further approaches. 

\bibliographystyle{abbrvnat}
\bibliography{bibliography}

\newpage
\appendix

\section{Tables}\label{app:tables}

\begin{table}[htbp]
\centering
\begin{tabular}{|c|r|r|r|r|r|r|r|}
\hline
\multicolumn{1}{|c|}{\textbf{}} &
  \multicolumn{1}{c|}{\textbf{Votes}} &
  \multicolumn{1}{c|}{\textbf{CCM}} &
  \multicolumn{1}{c|}{\textbf{BPM}} &
  \multicolumn{1}{c|}{\textbf{TPM}} &
  \multicolumn{1}{c|}{\textbf{TPM3}} &
  \multicolumn{1}{c|}{\textbf{TPP}} &
  \multicolumn{1}{c|}{\textbf{TPP3}} \\ \hline
\textbf{XP}   & 22.8   & 26.8  & 26.8  & 23.9  & 25.4  & 23.2  & 24.6  \\ \hline
\textbf{YQ}   & 20.7   & 20.3  & 20.3  & 21.7  & 23.2  & 21.0  & 22.5  \\ \hline
\textbf{LP}   & 18.3  & 19.6  & 19.6  & 19.6  & 20.3  & 18.8  & 19.6  \\ \hline
\textbf{YB}   & 16.0  & 18.1  & 18.1  & 16.7  & 18.1  & 16.7  & 17.4  \\ \hline
\textbf{INN}  & 8.7   & 8.0   & 8.0   & 8.7   & 9.4   & 8.7   & 9.4   \\ \hline
\textbf{XA}   & 3.8   & 0.0   & 0.0   & 3.6   & 3.6   & 3.6   & 3.6   \\ \hline
\textbf{ZR}   & 1.0   & 0.0   & 0.0   & 0.7   & 0.0   & 0.7   & 0.0   \\ \hline
\textbf{ZB}   & 0.8   & 0.0   & 0.0   & 0.7   & 0.0   & 0.7   & 0.0   \\ \hline
\textbf{YK}   & 0.8   & 0.7   & 0.7   & 0.7   & 0.0   & 0.7   & 0.0   \\ \hline
\textbf{YU}   & 0.8   & 0.0   & 0.0   & 0.7   & 0.0   & 0.7   & 0.0   \\ \hline
\textbf{ZL}   & 0.8   & 0.0   & 0.0   & 0.7   & 0.0   & 0.7   & 0.0   \\ \hline
\textbf{T}    & 0.7   & 0.7   & 0.7   & 0.7   & 0.0   & 0.7   & 0.0   \\ \hline
\textbf{YX}   & 0.7   & 0.0   & 0.0   & 0.7   & 0.0   & 0.7   & 0.0   \\ \hline
\textbf{ZW}   & 0.7   & 0.7   & 0.7   & 0.7   & 0.0   & 0.0   & 0.0   \\ \hline
\textbf{IND9} & 0.7   & 0.7   & 0.7   & 0.0   & 0.0   & 0.7   & 0.7   \\ \hline
\textbf{XI}   & 0.7   & 0.7   & 0.7   & 0.0   & 0.0   & 0.0   & 0.0   \\ \hline
\textbf{ZZ}   & 0.6   & 0.7   & 0.7   & 0.0   & 0.0   & 0.7   & 0.7   \\ \hline
\textbf{WB}   & 0.5   & 0.7   & 0.7   & 0.0   & 0.0   & 0.0   & 0.0   \\ \hline
\textbf{P}    & 0.4   & 0.7   & 0.7   & 0.0   & 0.0   & 0.0   & 0.0   \\ \hline
\textbf{A}    & 0.2   & 0.7   & 0.7   & 0.0   & 0.0   & 0.0   & 0.0   \\ \hline
\textbf{XM}   & 0.2   & 0.7   & 0.7   & 0.0   & 0.0   & 0.7   & 0.7   \\ \hline
\textbf{IND1} & 0.1   & 0.0   & 0.0   & 0.0   & 0.0   & 0.7  & 0.7  \\ \hline
\end{tabular}
\caption{Proportion of votes obtained by list, and proportion of seats obtained by least under each method.}
\label{tab:votes-bylist}
\end{table}

\begin{table}[htbp]
\centering
\begin{tabular}{|c|r|r|r|r|r|}
\hline
 &
  \multicolumn{1}{c|}{\textbf{CCM}} &
  \multicolumn{1}{c|}{\textbf{TPM}} &
  \multicolumn{1}{c|}{\textbf{TPM3}} &
  \multicolumn{1}{c|}{\textbf{TPP}} &
  \multicolumn{1}{c|}{\textbf{TPP3}} \\ \hline
\textbf{1}  & 0.22 & 0.24 & 0.24 & 0.28 & 0.28 \\ \hline
\textbf{2}  & 0.15 & 0.29 & 0.29 & 0.29 & 0.29 \\ \hline
\textbf{3}  & 0.23 & 0.24 & 0.32 & 0.20 & 0.13 \\ \hline
\textbf{4}  & 0.13 & 0.19 & 0.19 & 0.19 & 0.19 \\ \hline
\textbf{5}  & 0.06 & 0.14 & 0.14 & 0.14 & 0.14 \\ \hline
\textbf{6}  & 0.08 & 0.11 & 0.14 & 0.11 & 0.14 \\ \hline
\textbf{7}  & 0.08 & 0.08 & 0.08 & 0.08 & 0.08 \\ \hline
\textbf{8}  & 0.11 & 0.12 & 0.20 & 0.12 & 0.20 \\ \hline
\textbf{9}  & 0.12 & 0.20 & 0.20 & 0.17 & 0.17 \\ \hline
\textbf{10} & 0.09 & 0.10 & 0.17 & 0.10 & 0.17 \\ \hline
\textbf{11} & 0.14 & 0.14 & 0.27 & 0.14 & 0.27 \\ \hline
\textbf{12} & 0.11 & 0.11 & 0.16 & 0.11 & 0.16 \\ \hline
\textbf{13} & 0.32 & 0.19 & 0.19 & 0.19 & 0.19 \\ \hline
\textbf{14} & 0.08 & 0.18 & 0.18 & 0.18 & 0.18 \\ \hline
\textbf{15} & 0.10 & 0.22 & 0.21 & 0.26 & 0.21 \\ \hline
\textbf{16} & 0.08 & 0.08 & 0.08 & 0.08 & 0.08 \\ \hline
\textbf{17} & 0.08 & 0.18 & 0.18 & 0.18 & 0.18 \\ \hline
\textbf{18} & 0.17 & 0.26 & 0.17 & 0.17 & 0.17 \\ \hline
\textbf{19} & 0.15 & 0.18 & 0.18 & 0.18 & 0.18 \\ \hline
\textbf{20} & 0.09 & 0.17 & 0.20 & 0.20 & 0.20 \\ \hline
\textbf{21} & 0.11 & 0.20 & 0.11 & 0.20 & 0.11 \\ \hline
\textbf{22} & 0.28 & 0.28 & 0.18 & 0.28 & 0.18 \\ \hline
\textbf{23} & 0.09 & 0.11 & 0.09 & 0.11 & 0.15 \\ \hline
\textbf{24} & 0.20 & 0.21 & 0.21 & 0.21 & 0.21 \\ \hline
\textbf{25} & 0.28 & 0.17 & 0.17 & 0.17 & 0.17 \\ \hline
\textbf{26} & 0.13 & 0.23 & 0.23 & 0.23 & 0.23 \\ \hline
\textbf{27} & 0.16 & 0.29 & 0.29 & 0.29 & 0.29 \\ \hline
\textbf{28} & 0.24 & 0.25 & 0.25 & 0.24 & 0.24 \\ \hline
\end{tabular}
\caption{One-dimensional Gallagher index computed for each district.}
\label{tab:GI-by-dist}
\end{table}

\end{document}